%% file: main.tex
\documentclass[a4,11pt]{article}
\usepackage[top=.9in, bottom=.9in, left=.95in, right=.95in]{geometry}
\usepackage{amsmath,amsthm}
\usepackage{amssymb}
\usepackage{authblk}
\usepackage{paralist}
\usepackage[ruled]{algorithm}
\usepackage[noend]{algpseudocode}
\usepackage{hyperref}

\usepackage{graphicx}
\usepackage{caption}
\usepackage{subcaption}

\usepackage{epsfig}
\usepackage{enumerate}

\usepackage{chngcntr}
\usepackage{apptools}

\usepackage[final]{changes}

\newif\iffull
\fulltrue

\newif\ifshort
\shorttrue

\newif\ifrobocza
\roboczafalse

\newcommand{\comment}[1]{}
\usepackage{etoolbox}\AtBeginEnvironment{algorithmic}{\small}

\comment{
	\newtheoremstyle{redstyle}
	{3pt}
	{3pt}
	{\color{black}}
	{}
	{\color{red}\bfseries}
	{:}
	{.5em}
	{}
	\theoremstyle{redstyle} 
}

\algnewcommand{\LineComment}[1]{\State \(\triangleright\) #1}

\usepackage{color}

\DeclareMathOperator{\up}{ref}

\newcommand{\cluster}{\text{cluster}}

\newcommand{\gora}[1]{} 

\newtheorem{definition}{Definition}
\newtheorem{theorem}{Theorem}
\newtheorem{lemma}{Lemma}

\newtheorem{fact}{Fact}

\newtheorem{proposition}{Proposition}

\ifrobocza
\newcommand{\labell}[1]{\label{#1}\marginpar{#1}}
\newcommand{\todo}[1]{\noindent\colorbox{red}{todo: #1}}

\else
\newcommand{\labell}[1]{\label{#1}}
\newcommand{\todo}[1]{}

\newcommand{\cc}[1]{}
\fi

\newcommand{\prob}{\text{Prob}}
\newcommand{\NAT}{{\mathbb N}}

\newcommand{\dist}{\text{d}}

\newcommand{\cT}{{\mathcal T}}
\newcommand{\cN}{{\mathcal N}}

\newcommand{\eps}{\varepsilon}

\renewcommand{\thefootnote}{\fnsymbol{footnote}}

\newcommand{\clid}{\phi} 

\newcommand{\cB}{{\mathcal B}}

\newcommand{\calP}{{\cal P}}

\newcommand{\calT}{{\cal T}}

\newcommand{\E}{{\mathbb E\,}}

\newcommand{\remove}[1]{}

{\ $\blacksquare$ \smallbreak}

\newcommand{\lemmaApp}[2]{
{\noindent \textbf{Lemma \ref*{#1}.}}
\emph{#2}
}

\begin{document}

	\title{Deterministic Digital Clustering of Wireless Ad Hoc Networks
		\footnote{%
			This work was supported by the Polish National Science Centre 
			grant DEC-2012/07/B/ST6/01534.
		}
}
	
	
	\comment{ 
		\author{Tomasz Jurdzinski\inst{1} \and Dariusz R.~Kowalski\inst{2} \and Michał Rozanski\inst{1} \and Grzegorz Stachowiak\inst{1}}
		\authorrunning{T.~Jurdzinski, D.R.~Kowalski, M.~Rozanski, G.~Stachowiak} 
		\institute{Institute of Computer Science, University of Wroc{\l}aw, Poland
			\and
			Department of Computer Science, University of Liverpool,
			United Kingdom
		}
	} 
	
	\author[1]{Tomasz Jurdzi\'{n}ski}
	\author[2]{Dariusz R.~Kowalski}
	\author[1]{Micha\l~R\'{o}\.{z}a\'{n}ski}
	\author[1]{Grzegorz Stachowiak}

	\affil[1]{
	Institute of Computer Science, University of Wroc{\l}aw, Poland}
	\affil[2]{
			Department of Computer Science, University of Liverpool,
			United Kingdom
}
	
	\date{}
	
	\maketitle              
	
	\begin{abstract}
We consider deterministic distributed communication in wireless ad hoc 
networks of identical weak devices 
under the SINR model
without predefined infrastructure. 
Most algorithmic results in this 
model
rely on various
additional features or capabilities, e.g., randomization, access to geographic coordinates, power control, carrier sensing with various precision of measurements, 
and/or interference cancellation. 
We study a 
\emph{pure}
scenario, when no such properties are available. 

The key difficulty in the considered pure scenario stems from the fact that
it is not possible to distinguish successful delivery of message 
sent by
close neighbors from those 
sent by
nodes located on transmission boundaries. 
This problem has been avoided by appropriate model assumptions in many results,
which simplified algorithmic solutions.

As a general tool, we develop a deterministic distributed
\emph{clustering} algorithm, which splits nodes
of a multi-hop network into clusters such that: (i)~each cluster is included in a ball of constant diameter; (ii)~each ball of diameter $1$ contains nodes from $O(1)$ clusters. 
Our solution relies on a new type of combinatorial structures (selectors), which might be of independent interest.

Using the clustering, we develop a 
deterministic distributed \emph{local broadcast} algorithm 
accomplishing this task in
$O(\Delta\log^*N \log N)$ rounds,
where $\Delta$ is the density of the network.
To the best of
our knowledge, this is the first solution in pure scenario
which is only $\text{polylog}(n)$ away from the universal lower bound $\Omega(\Delta)$, valid also for scenarios with randomization and other features.
Therefore, none of these features substantially helps in performing the
local broadcast task.

Using clustering, we also build a deterministic {\em global broadcast} algorithm that terminates within $O(D(\Delta+\log^* N)\log N)$ rounds, where $D$ is the diameter of the network. This result is complemented by a lower bound $\Omega(D\Delta^{1-1/\alpha})$,  
where $\alpha>2$ is the path-loss parameter of the environment. 
This lower bound, in view of 
previous work,
shows that randomization or knowledge of own location substantially help (by a factor polynomial in $\Delta$) in {\em the global broadcast}. 
Therefore, unlike in the case of local broadcast, some additional model features
may help in global broadcast.

	\end{abstract}

	
	\comment{
	\small
		\noindent
		Corresp.\ author: {T.~Jurdzinski}, (\url{tju@cs.uni.wroc.pl})
		\noindent
		Regular paper, eligible for the best student paper award -- {M.~R\'{o}\.{z}a\'{n}ski} is a full time PhD student and has made
		a significant contribution. 
	}
	

	
	\renewcommand{\thefootnote}{\arabic{footnote}}
	
	

	
	



	\section{Introduction}
	\label{s:intro}
	\input{intro.tex}

	\section{Preliminaries}
	\labell{s:preliminaries}
	\input{preliminaries.tex}

	\section{Network sparsification and its applications}
	\labell{s:sparsification}
	\input{sparsification.tex}

	\section{Communication problems}\labell{s:comm:prob}
	\input{commProblems.tex}

	\section{Lower Bound}\labell{s:lower:bound}
	\input{lowerbound.tex}
	
	\section{Conclusions}
We have shown that it is possible to build efficiently a clustering of an ad hoc network without use of randomization in a very harsh model of wireless SINR networks. Using the clustering algorithm, we developed a very efficient local broadcast algorithm. On the other hand, by an appropriate lower bound, we  exhibited importance of randomization or availability of 
coordinates for complexity of the global broadcast problem.
The exact complexity of global broadcast remains open as well as the impact of other features, e.g., carrier sensing or power 
\ifshort
control.
\else
control in
the deterministic setting.
\fi

We developed a combinatorial structure called \emph{witnessed (cluster aware) strong selector} along with several communication primitives, 
\ifshort
which might be of independent interest.
\else
which allowed for improving the complexity of our algorithms, and might be of independent interest.
\fi
\comment{
in a very harsh 
In the paper we presented two solutions to broadcast and local broadcast with complexity, respectively $O(n\log^2 n), O(D(\Delta+\log^* n)\log n)$ and $O((\Delta+\log^* n)\log^2 n)$. We gave a lower bound $\Omega(D\Delta^{1-1/\alpha})$ for the broadcast problem which separates determinism from randomness in an ad hoc setting. The local broadcast algorithm is a first deterministic solution to the problem and achieves the complexity which is close to best known randomized solutions. Moreover, we developed a combinatorial structure called \emph{eliminating selector} along with several communication primitives, which allowed for improving the complexity of our algorithms, and are of independent interest. 
}
	


\input{main-bbl.tex}
	\newpage
	\section{Appendix}
	\input{appendix1.tex}

\end{document}

%% file: intro.tex
We 
study distributed algorithms 
in ad hoc wireless networks in the SINR \added{({\em Signal-to-Interference-and-Noise Ratio}) } model with uniform transmission power. We consider the \emph{ad hoc} setting where both capability and knowledge of nodes are limited -- nodes know only the basic parameters of the SINR model (i.e., $\alpha, \beta, \eps, \mathcal{N}$) and upper bounds $\Delta$, $n$ on the degree and the size of the network, such that the actual maximal degree is $O(\Delta)$ and the size is $n^{O(1)}$. Such a setting appears in networks without infrastructure of base nodes, access points, etc., reflecting e.g.\ scenarios where large sets
of sensors are distributed in an area of rescue operation, environment monitoring,
or in prospective internet of things applications.
Among others, we study basic communication problems as global and local broadcast, as well as primitives used as tools to coordinate computation and communication in the networks (e.g., leader election and wake-up). Most of these problems were studied in the model of graph-based radio networks over the years. Algorithmic study of the closer to reality SINR model has started much later. This might be caused by specific dynamics of interferences and signal attenuation, which seems to be more difficult for modeling and analysis than graph properties appearing in the radio network model (see e.g.\ \cite{KantorFOCS15}). 

As for the problems studied in this paper, several distributed algorithms have been presented in recent years. However, all these solutions were either randomized or relied on the assumption that nodes
of a network know their own coordinates in a given metric space. In contrast, the aim of this paper is to check how efficient could be solutions 
without randomization, availability of locations, power control,
carrier sensing, interference cancellation or other 
additional model
features, in the context of local and global communication problems.
\added{Our research objective is motivated twofolds. 
Firstly, examination of necessity of randomization is a natural research topic in algorithmic community for classic models of sequential computation as well as for models of 
parallel and distributed computing (see e.g.\ \cite{Bar-YehudaGI92} for
an example in a related radio networks model). Moreover, as wireless ad hoc networks are
usually built from computationally limited devices run on batteries, it is desirable
to use simple and energy efficient algorithms which do not need access to several
sensing capabilities or true randomness.}

\subsection{The Network Model.}
We consider a wireless network consisting of nodes located on the 2-dimensional Euclidean space\footnote{Results of this paper can be generalized to so-called bounded-growth metric spaces with the same asymptotic complexity bounds.}. We model transmissions in the network with the SINR constraints. The model is determined by fixed parameters: path loss $\alpha>2$, threshold $\beta>1$, ambient noise $\mathcal{N}>0$ and transmission power $\mathcal{P}$. 
The value of $SINR(v,u,\cT)$ for given nodes $u,v$ and a set of concurrently transmitting nodes $\cT$ 
is defined as
\begin{equation}\label{e:sinr}
SINR(v,u,\cT) = \frac{\mathcal{P}/d(v,u)^{\alpha}}{\cN+\sum_{w\in\cT\setminus\{v\}}\mathcal{P}/d(w,u)^{\alpha}},
\end{equation}
where $d(u,v)$ denotes the distance between $u$ and $v$.
A node $u$ successfully receives a message from $v$  
iff
$v\in \cT$ and $SINR(v,u,\cT)\ge\beta$, where $\cT$ is the set of nodes transmitting at the same time.
\emph{Transmission range} is the maximal distance at which a node can be heard provided there are no other transmitters in the network. Without loss of generality we assume that the transmission range is all equal to $1$. This assumption  implies that the relationship $\mathcal{P}=\mathcal{N}\beta$ holds. However, it does not affect generality and asymptotic complexity of presented results.

\noindent\textbf{Communication graph} 
In order to describe the topology of a network as a graph, we set a connectivity parameter $\eps\in(0,1)$.
The {\em communication graph} $G=(V,E)$ of a given network consists of all nodes and edges $\{v,u\}$ \replaced{connecting}{between} nodes that are in distance at most $1-\eps$\replaced{, i.e., $d(u,v)\le 1-\eps$.}{, where $\eps\in(0,1)$is a fixed constant (the connectivity parameter)}
Observe that a node can receive a message from a node that is not its neighbor in the graph, provided interference from other nodes is small enough.
The communication graph, defined as above, is a standard notion in the analysis of ad hoc multi-hop communication in the SINR model, cf., \cite{DaumGKN13,JKRS14}.

\noindent\textbf{Synchronization and content of messages}
We assume that algorithms work synchronously in rounds. In a single round, a node can transmit or receive a message from some node in the network and perform local computation. The size of a single message is limited to $O(\log N)$.
 

\ifshort
\noindent\textbf{Knowledge of nodes}
\else
\paragraph{Knowledge of nodes}
\fi
Each node has a unique identifier from the set \added{$[N]=\{1,\ldots,N\}$,} where $N=n^{O(1)}$ is the upper bound on the size $n$ of the network. Moreover, nodes know $N$, the SINR parameters -- $\mathcal{P},\alpha,\beta,\eps,\mathcal{N}$, 
the linear upper bounds on the diameter $D$ and the degree $\Delta$ of the communication graph. 

\ifshort
\noindent\textbf{Communication problems}
\else
\paragraph{Communication problems}
\fi 
The \emph{global broadcast} problem is to deliver a message from the designated source node $s$ to all the nodes in the network, perhaps through relay nodes as not all nodes are within transmission range of the source in multi-hop networks.
At the beginning of an execution of a global broadcast algorithm, only the source node $s$
is active (i.e., only $s$ can participate in the algorithm's execution).
A node $v\neq s$ starts participating in 
an execution of the algorithm only after receiving the first message from another
node. (This is so-called non-spontaneous wake-up model, popular in the literature.) 
In the \emph{local broadcast \added{problem}}, the goal is to make each node to successfully transmit its own message to its neighbors in the communication graph. 
Here, all nodes start participating in an algorithm's execution at the same
round.

\subsection{Related Work}
In the last years the SINR model was extensively studied, both from the perspective of its structural properties 
\cite{AronovK15, HalldorssonT15, KantorFOCS15} 
and algorithm design 
\cite{GoussevskaiaMW08,BarenboimSIROOCCO15,DaumGKN13, HalldorssonHL15,JK-DISC-12,YuHWL12,Kesselheim11,hobbs2012deterministic}. 

The first work on local broadcast in SINR model by Goussevskaia et al. \cite{GoussevskaiaMW08} presented an $O(\Delta\log n)$ randomized algorithm under assumption
that density $\Delta$ is known to nodes. After that, the problem was studied in various settings. Halldorsson and Mitra presented an $O(\Delta+\log^2n)$ algorithm in a model with feedback~\cite{HM12}. Recently, for the same setting Barenboim and Peleg presented solution working in time $O(\Delta + \log n\log\log n)$~\cite{BarenboimSIROOCCO15}.	For the scenario when the degree $\Delta$ is not known Yu et al. in \cite{YuHWL12} improved on the $O(\Delta\log^3 n)$ solution of Goussevskaia et al. to $O(\Delta\log n+ \log^2n)$. 

\added{In \cite{JK-DISC-12} a $O(\Delta\log^3 n)$-round deterministic local broadcast algorithm was presented under the assumption that nodes have access to location information (their coordinates).}
However, 
no deterministic algorithm for local broadcast was known in the scenario that
nodes do not know their coordinates.

\added{Previous results on local broadcast are collected and compared with our result in 
Table~\ref{tab:local}. There are also a few algorithms which require carrier sensing
or power control (e.g., \cite{FuchsW13}). As the model and techniques in this line
of work are far from
the topic of our paper, those solutions are not presented in the table.}

\begin{table}[h]
\begin{center}
\begin{tabular}{|l|c|c|c|}
  \hline 
  Paper & Model & Knowledge & Time\\
  \hline
	\multicolumn{4}{|c|}{Randomized algorithms}\\
	\hline
  \cite{GoussevskaiaMW08} & & $\Delta, n$ & $O(\Delta\log n)$ \\
  \cite{GoussevskaiaMW08} & &  $n$ & $O(\Delta\log^3 n)$ \\
	\cite{YuHWL12} & & $n$ & $O(\Delta\log n+ \log^2n)$\\
	\cite{HM12} & feedback & $\Delta,n$ & $O(\Delta+\log^2n)$ \\
	\cite{BarenboimSIROOCCO15} & feedback & $\Delta,n$ & $O(\Delta + \log n\log\log n)$ \\
  \hline
	\multicolumn{4}{|c|}{Deterministic algorithms (for $N=\text{poly}(n)$)}\\
	\hline
  \cite{JK-DISC-12} & location & $\Delta, N$ & $O(\Delta\log^3 n)$ \\
  This work & & $\Delta, N$ & $O(\Delta\log^* n \log n)$ \\
	\hline
\end{tabular}
\end{center}
\caption{Algorithms for the local broadcast problem. The column \emph{Model} enumerates
features which do not appear in the model of this paper, while they are needed to implement particular algorithms.}
\label{tab:local}
\end{table}
For the global broadcast problem a few deterministic solutions are known. 
\added{If the information about location of nodes is available, the global broadcast can be accomplished deterministically in time $O(D\log^2 n)$ 
\cite{JKS13}. }
\added{In contrast, if connectivity of a network might rely
on so-called weak links and the model does not allow for any
channel sensing nor access to geographic coordinates, deterministic global broadcast requires $\Theta(n\log N)$ rounds \cite{JRS17}, which gives logarithmic improvement
over $O(n\log^2n)$-time algorithm from \cite{DaumGKN13} for this model.}

The main randomized results on global broadcast in ad hoc settings include papers of Daum et al. \cite{DaumGKN13} and Jurdzinski et al. \cite{JKRS14}. Solutions with complexity, respectively $O((D \log n) \log^{\alpha+1} g)$ and $O(D\log^2 n)$ are presented, where $g$ is a parameter depending on the geometry of the network, \added{at most exponential with
respect to $n$}. \added{(Results on global broadcast are collected and compared with our result in 
Table~\ref{tab:global}.)}

Recently Halldorsson et al.\ \cite{HalldorssonHL15} proposed an algorithm which 
can be faster assuming that nodes are equipped with some extra capabilities
(e.g., detection whether a received message is sent from a close neighbor)
and the interference function on the local (one-hop) level is defined as in the classic radio network model. 
\added{Given additional model features as channel/carrier sensing,
total interference estimation by signal measurements and others,
efficient algorithms for various communication problems have been
designed in recent years, e.g., 
\cite{Gudm,RichaS16,OgiermanRSSZ14}.}

\begin{table}[h]
\begin{center}
\begin{tabular}{|l|c|c|c|c|}
  \hline 
  Paper & Model: & Model: & Knowledge & Time\\
  & extra features & weaknesses &  & \\
  \hline
	\multicolumn{5}{|c|}{Randomized algorithms}\\
	\hline
	\cite{DaumGKN13} & &  & $n$ & $O((D \log n) \log^{\alpha+1} g)$ \\
	\cite{DaumGKN13} & & weak links & $n$ & $O(n\log^2 n)$\\
	\cite{JKRS13} & location & & $n$ & $O(D\log n+\log^2n)$\\
	\cite{JKRS14} & &  & $n$ & $O(D\log^2 n)$ \\
  \hline
	\multicolumn{5}{|c|}{Deterministic algorithms (for $N=\text{poly}(n)$)}\\
	\hline
  \cite{JKS13} & location & & $N$ & $O(D\log^2 n)$ \\
	\cite{JRS17} & & weak links & $N$ & $\Theta(n\log n)$\\
  This work & & & $\Delta, N$ & $\Omega(D\Delta^{1-1/\alpha})$ \\
  This work & & & $\Delta, N$ & $O(D(\Delta+\log^*n)\log n)$ \\
	\hline
\end{tabular}
\end{center}
\caption{Algorithms for the global broadcast problem.}
\label{tab:global}
\end{table}

In the related multi-hop radio network model the global broadcast problem is well examined. The complexity of randomized broadcast is of order of 
$\Theta((D + \log n) \log(n/D))$~\cite{Bar-YehudaGI92,CzumajRytter-FOCS-03,KowalskiPelc2003}. 
A series of papers on deterministic global broadcast give the upper bound of $O(n\min\{\log n, \log D\log\log(D\Delta/n)\}$ 
\ifshort
\cite{ChlebusGGPR00,CzumajRytter-FOCS-03,DeMarco-SICOMP-10} 
\else
\cite{ChlebusGGPR00,MarcoP01,CzumajRytter-FOCS-03,DeMarco-SICOMP-10,KowalskiP04} 
\fi
with the lower bound of 
\comment{
$\Omega(n\log D)$~\cite{ClementiMS01} and 
}
$\Omega(n\log n/\log(n/D))$~\cite{KowalskiPelc2003}. In terms of $\Delta$ and $D$, the best bounds are 
$O(D\Delta\log(N/\Delta)\log^{1+\alpha}N)$ and $\Omega(D\Delta\log(N/\Delta)$ \cite{ClementiMS03}.
For the closest to our model geometric unit-disk-graph radio networks,
the complexity of broadcast is $\Theta(D\Delta)$ \cite{EmekGKPPS09,EmekKP16}, even when nodes know their coordinates.

\added{Recently, communication algorithms have been also designed in a very weak
beeping networks \cite{CornejoK10,ForsterSW14}, where a node can only distinguish between
$0$ or at least one of its neighbors transmitting on the communication channel.}

\subsection{Our Contribution}
As a general tool, we develop a clustering algorithm which splits nodes
of a multi-hop network into clusters such that: (i)~each cluster is included in a ball of constant diameter; (ii)~each ball of diameter $1$ contains nodes from $O(1)$ clusters. Using the clustering algorithm, we develop a \emph{local broadcast} algorithm\deleted{ in this setting} for ad hoc wireless networks, which accomplishes the task in $O((\Delta+\log^* n)\log n)$ rounds,
where $\Delta$ is the density of the network.
Up to our knowledge, this is the first solution in the considered pure ad hoc scenario
which is only $\text{polylog}(n)$ away from the trivial universal lower bound $\Omega(\Delta)$, valid also for scenarios with randomization and other features.
%

Using clustering, we also build a deterministic global broadcasting algorithm that terminates within $O(D(\Delta+\log^* n)\log n)$ rounds, where $D$ is the diameter of the network graph. This result is complemented by a lower bound $\Omega(D\Delta^{1-1/\alpha})$ in the network of degree $\Delta$, where $\alpha>2$ is the path-loss parameter of the environment. 
This lower bound, in view of previous work, shows that randomization or knowledge of own location substantially help (by a factor polynomial in $\Delta$) in {\em the global broadcast}. 
%
Previous results on global broadcast in related models achieved time $O(D\,\text{polylog} N)$, thus they were independent of the networks density $\Delta$. They relied, however, on either randomization, or access to coordinates of nodes in the Euclidean space, or carrier sensing (a form of measurement of strength of received signal). Without these capabilities, as we show, the polynomial dependence of $\Delta$ is unavoidable.

Therefore, our results prove that additional model features may help in global
communication problems, but not much in local problems such as local broadcast,
in which advanced algorithms are (almost) equivalent to the presence
of additional model features.
Using designed algorithmic techniques we also provide efficient solutions to the wake-up problem and the global leader election problem.

\comment{
\begin{table}[h]
\begin{tabular}
\end{tabular}
\caption{Local broadcast algorithms}
\end{table}
}

\subsection{High Level Description of Our Technique}
The key challenge in designing efficient algorithm in the ad hoc wireless scenario is the interference from dense areas of a network. 
Note that if randomization was available,
nodes could adjust
their transmission probabilities and/or signal strength such that the expected interference from each
ball of radius $1$ is bounded by a constant. 

Another problem stems from the fact that it is impossible to distinguish
received messages which are sent by close neighbors from those 
sent by
nodes on boundaries of the transmission range.
This issue could be managed if nodes had access to their geographic coordinates or were equipped with carrier sensing capabilities (thanks to them, the distance of a neighbor can be estimated based on a measurement of the strength
of received signal).
In the pure ad hoc model considered in this paper, all those tools are not
available.

Our algorithmic solutions rely on the notion of $r$-\emph{clustering}, i.e., a partition of a set of nodes into clusters such that: (i)~each cluster is included in a ball of 
\added{radius} $r$; (ii)~each ball of \added{radius} $1$ contains nodes from $O(1)$ clusters and (iii)~each node knows its cluster ID. 
First, we develop tools for (partially) clustered networks. 
We start with a \emph{sparsification} algorithm, which, 
given a set \added{of} clustered nodes $W$, gradually decreases the largest number of nodes in a cluster and eventually ends with a set $W'\subset W$ such that  $O(1)$ (and at least one!) nodes from each cluster of $W$ belongs to $W'$.
Using sparsification algorithm, we develop a tool for \emph{imperfect labeling} of clusters, which results in assigning temporary IDs (tempID)
in range $O(\Delta)$ to nodes
such that $O(1)$ nodes in each cluster have the same tempID.
Moreover, an efficient radius reduction algorithm is presented, which transforms $r$-clustering \added{for $r>1$} into $1$-clustering.


Two communication primitives are essential for efficient implementation of our tools, namely, 
\emph{Sparse Network Schedule} (SNS) and \emph{Close Neighbors Schedule} (CNS). 
%
Sparse Network Schedule is a communication protocol of length $O(\log N)$ which guarantees that, given an arbitrary set of nodes $X$ with constant density, each element $v$ of $X$ performs local broadcast (i.e., there is a round in which the message transmitted by $v$ is received in distance $1-\eps$). 
Close Neighbors Schedule is a communication protocol of length $O(\log N)$ which guarantees that, given an arbitrary set of nodes $X$ of density $\Gamma$ with 
$r$-clustering of $X$ for $r=O(1)$, each close enough pair of elements of $X$ from each
sufficiently dense cluster can hear each other during the schedule. 

Given the \added{above described} tools for clustered set of nodes, we build a global broadcasting algorithm, which works in phases. In the $i$th phase we assure that all
nodes awaken\footnote{A node is awaken in the phase $j$ if it receives the broadcast message for the first time in that phase.} in the $(i-1)$-st phase perform local broadcast. In this way, the set of nodes awaken in the first $i$
phases contains all nodes in communication graph of distance $i$ from the source. Moreover, 
we assure that all nodes awaken in a phase are $1$-clustered. 
(We start with the cluster formed by nodes in distance $\le 1$ from the source $s$, awaken in a round in which $s$ is the unique transmitter.)
%
A phase consists of three stages. In Stage~1, an \emph{imperfect labeling} of each cluster is done. 
In Stage~2, Sparse Network Schedule is executed $\Delta$ times\footnote{Here, we assume that $\Delta$ is the maximal number of nodes in a ball of radius $1$. This number differs from the degree of the communication graph by a constant multiplicative factor.}. A node with label $j$ participates in the $j$th execution of SNS only. 
In this way, all nodes transmit successfully on distance $1-\eps$. All nodes
awaken in Stage~2 inherit cluster ID from nodes which awaken them. In this
way, we have $2$-clustering of all nodes awaken in Stage~2. In Stage~3, a $1$-clustering of awaken nodes is formed by using an efficient algorithm that reduces the radius of clustering. 
\added{Figure~\ref{fig:global} provides an example describing a phase.}

\comment{
\begin{figure}[h]
	\centering
	\begin{subfigure}[t]{0.32\textwidth}
		\includegraphics[width=1.0\textwidth]{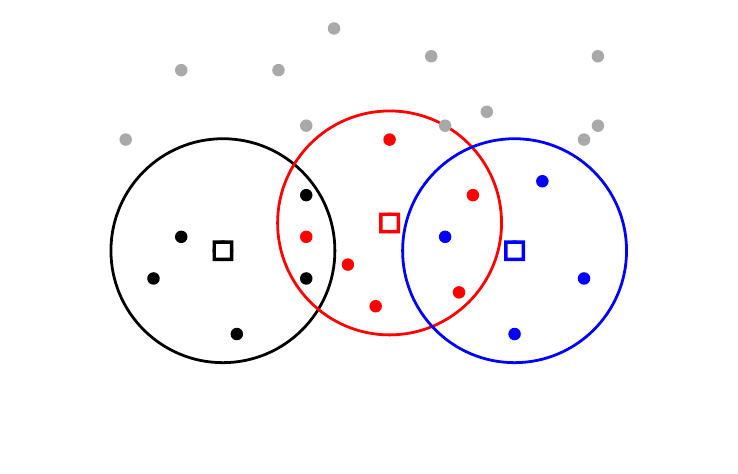}
		\caption{Black, red and blue nodes were awaken in phase $i-1$, they are $1$-clustered (colors correspond to clusters). Gray nodes are asleep neighbors of awaken nodes.
}
	\end{subfigure}
	\hfill
	\begin{subfigure}[t]{0.32\textwidth}
		\includegraphics[width=1.0\textwidth]{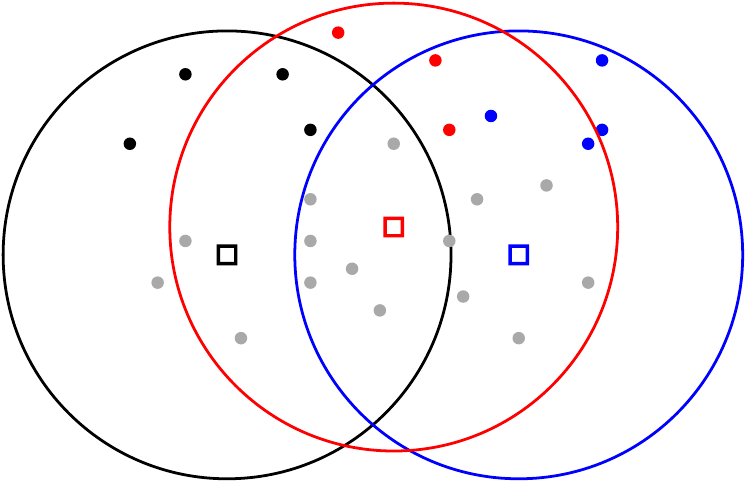}
		\caption{In Stage 2, asleep nodes (gray in Figure~(a); black, red and blue in Figure~(b)) receive messages from their awaken neighbors. After Stage 2 of phase $i$, newly awaken nodes inherit cluster IDs (colors) of
		neighbors which awaken them.}
	\end{subfigure}
	\hfill
	\begin{subfigure}[t]{0.32\textwidth}
		\includegraphics[width=1.0\textwidth]{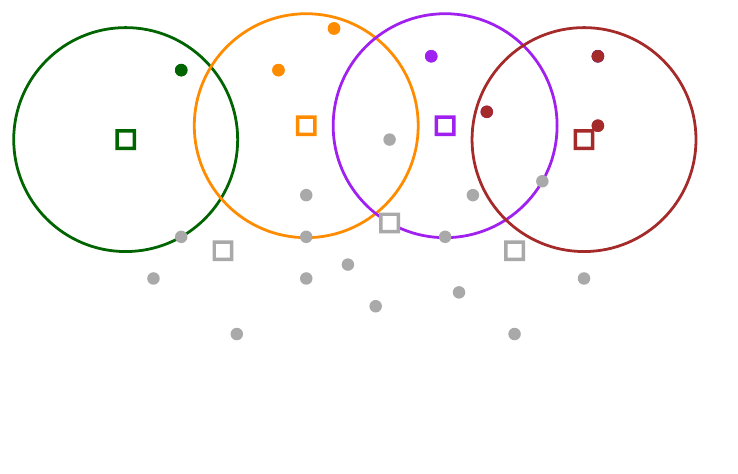}
		\caption{In Stage~3, newly awaken nodes build a new $1$-clustering, using the algorithm which
		``reduces'' a $2$-clustering to a $1$-clustering.}
	\end{subfigure}
	\caption{An illustration of phase $i>1$ of the global broadcast algorithm. Squares denote centres of clusters (corresponding to nodes from these clusters).
	}
	\label{fig:global}
\end{figure}
}
\begin{figure}[h]
	\centering
	\begin{subfigure}[t]{0.45\textwidth}
		\includegraphics[width=\textwidth]{globalPhase1.pdf}
		\caption{Black, red and blue nodes were awaken in phase $i-1$, they are $1$-clustered (colors correspond to clusters). Gray nodes are asleep neighbors of awaken nodes.
}
	\end{subfigure}
	\hfill
	\begin{subfigure}[t]{0.45\textwidth}
		\includegraphics[width=1.0\textwidth]{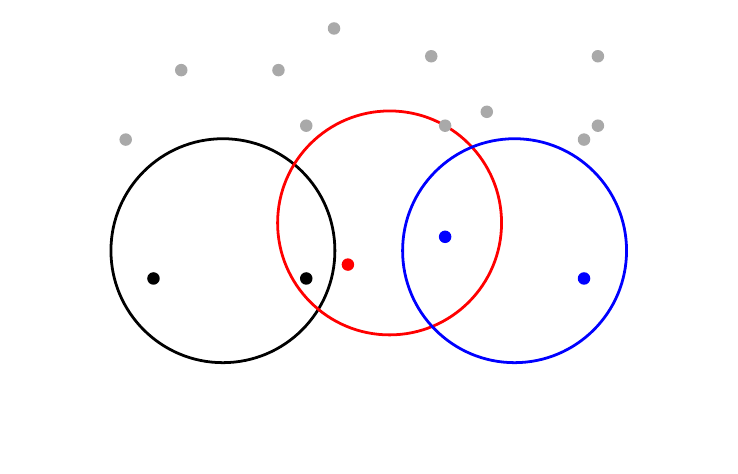}
		\caption{Stage 2: in the $j$th execution of Sparse Network Schedule, only nodes with label $j$ from each cluster participate (black, red and blue nodes).}
	\end{subfigure}
	\\
	\begin{subfigure}[t]{0.45\textwidth}
		\includegraphics[width=1.0\textwidth]{globalPhase2.pdf}
		\caption{During Stage 2 asleep nodes (gray in Figure~(a); black, red and blue in Figure~(c)) receive messages from their awaken neighbors. After Stage 2, newly awaken nodes inherit cluster IDs (colors) of
		neighbors which awaken them.}
	\end{subfigure}	
	\hfill
	\begin{subfigure}[t]{0.45\textwidth}
		\includegraphics[width=1.0\textwidth]{globalPhase3.pdf}
		\caption{In Stage~3, newly awaken nodes build a new $1$-clustering, using the algorithm which
		``reduces'' a $2$-clustering to a $1$-clustering.}
	\end{subfigure}
	\caption{An illustration of phase $i>1$ of the global broadcast algorithm. Squares denote centres of clusters (corresponding to nodes from these clusters).
	}
	\label{fig:global}
\end{figure}
Our algorithm for \emph{local broadcast} builds a $1$-clustering of the whole network, assigns tempIDs to nodes in clusters with use of the imperfect labeling algorithm, and eventually performs local broadcast by applying
Sparse Network Schedule for each prospective tempID separately.
(Recall that, in the local broadcast problem, all nodes are awaken simultaneously, just at the beginning of a protocol.) 
Thus, the key challenge here is to build a $1$-clustering starting from
an unclustered network.
First, we 
use our sparsification technique to build a sparse set of 
\emph{leaders} and a schedule $S$ such that each node of the network is in distance of $O(\log N)$ hops from
some leader with respect to the schedule $S$. Then, starting from clusters containing neighbors of leaders, we gradually build clustering of the whole network.

For efficient implementation of our solutions, we build a new type of selectors, called \emph{witnessed (cluster aware) strong selectors}. These selectors implement algorithmically an implicit collision detection, which filters out most connections on large distance in an execution
of Close Neighbors Schedule. Therefore, properties 
of the new selectors
might be applicable for designing efficient (deterministic)
solutions for other communication problems.

\subsection{Organization of the paper}
Section~\ref{s:preliminaries} contains basic notions and definitions. Combinatorial tools are introduced and applied to build SINR communication primitives in Section~\ref{s:comb:tools:sinr}.
Section~\ref{s:sparsification} describes the sparsification technique which leads to the clustering algorithm.
In Section~\ref{s:comm:prob}, we describe solutions of global/local broadcast 
as well as other communication problems. Finally, a lower bound for global broadcast is provided.

%% file: preliminaries.tex
\noindent\textbf{Geometric definitions}
Let $B(x,r)$ denote the ball of radius $r$ around point $x$ on the plane. 
That is, $B(x,r)=\{y\,|\, d(x,y)\le r\}$.
We identify $B(x,r)$ with the set of nodes of the network that are located inside $B(x,r)$. A \emph{unit ball} is a ball with radius $1$.
\deleted{Let $\chi(r_1,r_2)$ denote the maximal number of points which can be located in a ball of radius $r_1$ such  that each pair of points from the
set isin distance larger or equal to $r_2$. }
\added{Let $\chi(r_1,r_2)$ denote the maximal 
size of a set of points $S$ located in a ball of radius $r_1$ such  that $d(u,v)\ge r_2$ 
for each $u,v\in S$ such that $u\neq v$.
 }

\noindent\textbf{Clustered and unclustered sets of nodes}
Assume that each node from a given set is associated
with a pair of numbers $(v, \clid)$, where $v\in [N]$ is its unique ID and 
$\clid\in [N]$ is the \emph{cluster} of $v$, $\clid=\cluster(v)$.
A set $X$ of pairs $(v,\clid)$ associated to nodes is called 
a \emph{clustered set} of nodes. The set $X$ is partitioned into clusters,
where \emph{cluster} $\clid$ denotes the set $\{(v,\clid')\in X\,|\,\clid'=\clid\}$. 

An \emph{unclustered} set $X$ of nodes is just a subset of $[N]$, which might be also considered as a clustered set, where each node's cluster
ID is equal to $1$, i.e., $\cluster(v)=1$ for each $v\in X$.

\noindent\textbf{Geometric clusters}
Consider a clustered set of nodes $X$ such that each node is located
in the Euclidean plane (i.e, it has assigned coordinates determining its location).
The clustering of $X$ 
is 
an $r$-clustering for \added{$r\ge 1$} if the following conditions are satisfied: 
\begin{itemize}
\gora{4}
	\item For each cluster $\clid$, all of nodes from the cluster $\clid$
	are located in
	the ball $B(x,r)$, 
	where $x$ is an element of $\clid$
	called the \emph{center} of $\clid$.
	
\item
For each clusters $\clid_1,\clid_2$, 
the centers $x_1,x_2$ of $\clid_1, \clid_2$ are located in distance
at least $1-\eps$, i.e.,
$d(x_1,x_2)\ge 1-\eps$.
\end{itemize}

The \emph{density} $\Gamma$ of an unclustered network/set denotes the largest number of nodes in a unit ball.
For a clustered network/set $G$, the \emph{density} is equal to the largest number of elements of a cluster.
Below, we characterize dependencies between the degree of the communication graph and the density of the network.
\begin{fact}\labell{f:density}
Let $\Delta_V$ be equal to the largest degree in the communication graph $G$. Then,

\noindent 1.~There exist constants $0<c_1<c_2$ such that 
the density of each
unclustered network $V$ is in the range $[c_1\Delta_V,c_2\Delta_V]$.

\noindent 2.~There exist constants $0<c_1<c_2$ such that 
the density of each
$r$-clustered network $V$ for $r\in[1,2]$ is in the range $[c_1\Delta_V,c_2\Delta_V]$.
\end{fact}
As the density of a network and the degree of its communication graph are linearly dependent, we will 
use $\Delta$/$\Gamma$ to denote both the density and the degree of the communication graph.\footnote{This overuse of notations has no impact on asymptotic upper bounds on complexity of algorithms designed in this paper, since dependence on density/degree will be at most linear.}
For a clustered set of density $\Gamma$, a cluster is \emph{dense} when it contains at least $\Gamma/2$ elements. Similarly, for an unclustered set of density $\Gamma$, a unit ball $\cB$ is \emph{dense} when it contains at least $\Gamma/2$ elements.

In the following, we define the notion of a \emph{close pair}, essential for our network 
\emph{sparsification} technique. 
\deleted{Let $d_{\Gamma,r}$ be the number satisfying the equation $\chi(r, d_{\Gamma,r})=\Gamma/2$.}
\added{Let $d_{\Gamma,r}$ be the smallest number satisfying the inequality $\chi(r, d_{\Gamma,r})\ge\Gamma/2$.}
\added{Thus, according to the definition of the function $\chi$ and the
definition of a dense cluster/ball, 
the smallest distance between nodes of a dense cluster (unit ball, resp.) in an $r$-clustered (unclustered, resp.)
network is at most $d_{\Gamma,r}$ ($d_{\Gamma,1}$, resp.).}


\begin{definition}
	Nodes $u,w$ form a \emph{close pair} in an $r$-clustered network of density $\Gamma$ if the following conditions hold for some $\zeta\in (0,1]$ and cluster $\clid$:
	\begin{enumerate}[a)]
		\gora{4}
		\item 
		$\cluster(u)=\cluster(w)=\clid$, 
		\label{d:cp:a}
		\gora{2}
		\item $d(u,w) = \zeta d_{\Gamma,r}\le 1-\eps$, 
		\label{d:cp:b}
		\gora{2}
		\item
		$d(u,x)\ge d(u,w)$ and $d(w,x)\ge d(u,w)$ for each $x\not\in\{u,w\}$ from 
		the cluster $\clid=\cluster(u)=\cluster(w)$;
		\label{d:cp:c}
		\item $d(u',w')\ge d(u,w)/2$ for any $u'\neq w'$ such
		that $u',w'\in B(u,\zeta)\cup B(w,\zeta)$ and
		$\cluster(u')=\cluster(w')=\clid$.
		\label{d:cp:d}
	\end{enumerate}
The nodes $u,w$ form a \emph{close pair} in an unclustered network of density $\Gamma$ iff the above conditions (a)--(d) hold provided
$\cluster(x)=1$ for each 
node $x$ of a network and $r=1$.
%

A node $v$ is a \emph{close neighbor} of 
a node 
$u$ iff $(u,v)$ is a close pair.
\end{definition}

An intuition behind the requirements of the definition of a close pair is as follows.
In the clustered case, only the pairs
inside the same cluster are considered to be a close pair, by (\ref{d:cp:a}).
The requirement of (\ref{d:cp:b}) states that $u$ and $w$ can form
a close pair only in the case that $d(u,w)$ is at most the 
upper bound on the smallest distance
between closest nodes of a dense cluster/ball. 
The item (\ref{d:cp:c}) assures that $u$ is the closest node to $w$ and
$w$ is the closest node $u$. 
Finally, (\ref{d:cp:d}) states that $u$ and $w$ can be a close pair only
in the case that the distances between nodes in their close neighbourhood
are not much smaller than $d(u,w)$. This requirement's goal is to 
count $u,w$ as a close pair only in the case that $d(u,w)$ is
of the order of the smallest distance between nodes in some
neighbourhood.
Polynomial attenuation of strength of signals with the power
$\alpha>2$ in the SINR model will assure that, for a close pair $u, w$, a successful reception of a message
from $u$ to $w$ will depend on behaviour of some constant number
of the nodes which are closest to $u$ and $w$.
We formalize this intuition in further part of this work (Lemmas~\ref{l:close} and \ref{l:close:clustered}). In the following lemma we observe presence
of close pairs in dense areas of a network.

\newcommand{\ldensityclose}{Assume that the density of a set of nodes $X$ is $\Gamma$. Then,
\begin{enumerate}
\gora{4}
\item
If $X$ is unclustered, then there is a close pair in each ball
$B(x,5)$ such that $B(x,1)$ is dense. 
\gora{4}
\item
If $X$ is clustered, then there is a close pair in each dense cluster.
\end{enumerate}
}
\begin{lemma}\labell{l:density:close}
\ldensityclose
\end{lemma}

\noindent\textbf{Imperfect labeling}
Assume that a 
clustering of a set $X$ of nodes 
with
density $\Gamma$ 
is given.
Then, $c$-\emph{imperfect labeling} of $X$ is a labeling of all elements 
of $X$ such that label$(x)\le \Gamma$ 
 for each $x\in X$ and,
for each cluster $\clid\in[N]$ and each label $l\in[N]$, the number
of nodes from the cluster $\clid$ with label $l$ is at most $c$,
where $c$ is a fixed constant.
That is, the labeling is ``imperfect'' in the sense that, instead of unique labels, we guarantee
that each label is assigned to at most $c$ nodes in a cluster.

\section{Combinatorial Tools for SINR communication}
\labell{s:comb:tools:sinr}
In this section, we introduce combinatorial structures applied in our algorithms and
basic communication primitives using these structures.
\subsection{Combinatorial tools}
In this section we will apply the idea (known e.g.\ from
deterministic distributed algorithms in radio networks) to
use families of sets with specific combinatorial
properties as communication
protocols in such a way that the nodes from the $i$th set of the family
are transmitters in the $i$th round. Below, we give necessary
definitions to apply this approach, recall some results and
build our new combinatorial structures.

A \emph{transmission schedule} for unclustered (clustered, resp.) sets is defined by (and identified with) a sequence $S=(S_1,...,S_t)$ of subsets of $[N]$ ($[N]\times [N]$, resp.), where the $i$th set determines nodes transmitting in the $i$th round of the schedule. That is, a node with ID $v\in[N]$ (and $\cluster(v)=\clid$, resp.) transmits in round $i$ if and only if $v\in S_i$ ($(v,\cluster(v))\in S_i$, resp.). 

A set $S\subseteq[N]$ \emph{selects} $x\in X$ from $X\subseteq[N]$ when $S\cap X=\{x\}$.
The family $\mathcal{S}$ of sets over $[N]$ is called $(N,k)$-strongly selective family (or $(N,k)$-ssf) if for each subset $X\subseteq [N]$ such that $|X|\le k$ and each $x\in X$ there is a set $S\in\mathcal{S}$ that selects $x$ from $X$. It is well known that there exist $(N,k)$-ssf
of size $O(k^2\log (N/k))$ for each $k\le N$, 
\ifshort
\cite{ClementiMS01}.
\else
\cite{ClementiMS01}.
\fi

%


%


We introduce a combinatorial structure generalizing ssf in two ways.  
Firstly, it will take clustering into account, assuming that some clusters might be ``in conflict''.
Secondly, selections of elements from a given set $X$ will be ``witnessed'' by all nodes outside
of $X$.
As we show later, this structure helps to \added{build a sparse graph in a SINR network, such that each close pair is connected by an edge.}
\deleted{determine close pairs of nodes efficiently in a SINR network. }
We start from a basic variant called \emph{witnessed strong selector (wss)}, which does not take 
clustering into account. 
(A restricted variant of wss has been recently presented in \cite{JRS17}.)
Then, we generalize the structure to clustered sets.
We call this variant \emph{witnessed cluster aware strong selector (wcss).}

\ifshort
\noindent\textbf{Witnessed strong selector.}
\else
\paragraph{Witnessed strong selector.}
\fi
\added{Below, we introduce the } witnessed strong selector which extends the notion of ssf by requiring that, for each element $x$ of a given set $X$ of size $k$ and each $y\not\in X$, $x$ is selected from $X$ 
by such a set $S$ from the selector that $y\in S$ as well.

A sequence $\mathcal{S}=(S_1,\ldots,S_m)$ of sets over $[N]$ satisfies 
{\em witnessed strong selection property for a set $X\subseteq[N]$}, if 
for each $x\in X$ and $y\not\in X$ there is a set $S_i\in \mathcal{S}$ such that $X\cap S_i=\{x\}$ and $y\in S_i$. One may think
that $y$ is a ``witness'' of a selection of $x$ in such a case.
A sequence $\mathcal{S}=(S_1,\ldots,S_m)$ is a {\em $(N,k)$-witnessed strong selector} (or $(N,k)$-wss) of size $m$ if, for every subset $X\subseteq [N]$ of size $k$, the family $\mathcal{S}$ satisfies the witnessed strong selection property for $X$. 

Note that any $(N,k)$-wss is also, by definition, an $(N,k)$-ssf. 
Additionally, $(N,k)$-wss guarantees that each element outside of
a given set $X$ of size $k$ has to be a ``witness'' of selection of every element
from $X$.

Below we state an upper bound on the optimal size of $(N,k)$-wss. It is presented without a proof, since its generalization is given in Lemma~\ref{l:gen_el_selector} and accompanied by a proof which implies Lemma~\ref{l:comb:struct} (Lemma~\ref{l:comb:struct} is an
instance of
Lemma~\ref{l:gen_el_selector} for the number of clusters $l=1$).
\begin{lemma}\labell{l:comb:struct}
For each positive integers $N$ and $k\leq N$, there exists an $(N,k)$-wss of size $O(k^3\log N)$.
\end{lemma}

Now, we
generalize the notion of $(N,k)$-wss to the situation that witnessed strong selection property is analyzed for each cluster separately, assuming that each cluster might be in \emph{conflict} with $l$ other clusters. 
The conflict between the clusters $\clid_1$ and $\clid_2$ means that a selection of an element 
from $\clid_1$ by a set $S$ from the selector is possible only in the case that the considered set $S$ 
does not contain any element from the cluster $\clid_2$.

\noindent{\bf $(N,k,l)$-witnessed cluster aware strong selector.}
We say that a set $S\subseteq[N]^2$ is \emph{free} of cluster $\clid$ if for all $(x,\clid')\in S$ we have $\clid'\neq \clid$. A set $S$ is free of the set of clusters $C$ if it is free of each cluster $\clid\in C$. 
%
\comment{ 
Given a set $A\subset [N]^2$, 
$$\pi_1(A)=\{a\,|\, \exists_{b}\ (a,b)\in A\}$$
is the projection of $A$ on its first dimension.
For a family of sets $\mathcal{A}$, $\pi_1(\mathcal{A})$ is the family
of projections of elements of $\mathcal{A}$ on the first dimension.
} 
%
Let $X\subseteq [N]\times \{\clid\}$ be a set of nodes from the cluster $\clid$ and $C\subseteq [N]\setminus\{\clid\}$ be a set
%
of clusters in conflict with the cluster $\clid$.
Then, a sequence $\mathcal{S}=(S_1,...,S_m)$ of subsets of $[N]^2$ satisfies \emph{witnessed cluster aware selection
property (wcss property)} for $X$ with respect to $C$ if for each $x\in X$ and each $y\not\in X$ from cluster $\clid$ 
(i.e., $y\in [N]\times \{\clid\}$)
there is a set $S_i$ such that $S_i\cap X=\{x\}$, $y\in S_i$ and $S_i\cap \left([N]\times C\right)=\emptyset$
(i.e., $S_i$ is free of clusters from $C$).
In other words, wcss property requires that for each $x\in X$ 
and each $y\not\in X$ from $\clid=\cluster(x)$, $x$
is selected by some $S_i$,
$y$ is a witness of a selection of $X$ by $S_i$ (i.e., $y\in S_i$), 
and $S_i$ is free of the clusters from $C$.

A sequence $\mathcal{S}=(S_1,...,S_m)$ of subsets of $[N]^2$
is a $(N,k,l)$-witnessed clusters aware strong selector
(or $(N,k,l)$-wcss) if for any set of clusters $C\subseteq[N]$ of size $l$, any cluster $\clid\not\in C$ and any set $X\subseteq[N]\times\{\clid\}$ of size $k$, $\mathcal{S}$ satisfies wcss property for $X$
with respect to $C$.
%

\begin{lemma}
\label{l:gen_el_selector}
For each natural $N$, and $k,l\leq N$, there exists a $(N,k,l)$-wcss of size $O((k+l)lk^2\log N)$.
\end{lemma}
\begin{proof}
We use the probabilistic method \cite{AlonSpencer-book}.
Let $\mathcal{S}=(S_1,...,S_m)$, \added{when the length $m$ of the sequence will
be specified later,}  be a sequence of sets build in the following way.
The set $S_i$ is chosen \added{independently at random} as follows.
First, the set $C_i$ of ``allowed'' clusters is determined by 
adding each cluster ID from $[N]$ to $C_i$ with probability $1/l$. Then, the set $S_i$ is determined 
by adding independently $(x,\clid)$ to $S_i$ for each $x\in[N]$ and $\clid\in C_i$ 
to the set $S_i$, with probability $1/k$.
%
Let $\mathcal{T}$ be the set of tuples $(X,C,\clid,x,y)$
such that $X, C\subset[N]$, $|X|=k$, $|C|=l$, $\clid,x,y\in [N]$,
$\clid\not\in C$, $x\in X$, $y\not\in X$.
The size of $\mathcal{T}$ is 
$$|\mathcal{T}|= {N\choose k}{N\choose l}(N-l)k(N-k)<N^{k+l+3}.$$ 
For 
a fixed tuple $(X,C,\clid,x,y)\in\mathcal{T}$, let $\E_i$ be the conjunction
of three independent events: 
\begin{itemize}
\item
$\E_{i,0}$: $\clid\in C_i$ and $C\cap C_i=\emptyset$.
\item
$\E_{i,1}$: $(X\times\{\clid\})\cap S_i=\{(x,\clid)\}$,  
\item
$\E_{i,2}$: $(y,\clid)\in S_i$,
\end{itemize}
Then, $\E_i$ occurs with probability
$$p=\prob(\E_i)=\prob(\E_{i,0})\cdot\prob(\E_{i,1})\cdot\prob(\E_{i,2})=\frac1{l}\left(1-\frac1{l}\right)^l\cdot\frac1{k}\left(1-\frac1{k}\right)^{k-1}\cdot\frac1{k}
=\Omega\left(\frac1{lk^2}\right).$$
The sequence $S_1,S_2,\ldots, S_m$ is a $(N,k)$-wcss when
the event $\E_i$ occurs for each element of $\mathcal{T}$ for some index $i\in[m]$.
The probability that, for all indices $i\in[m]$, $\E_i$ does not occur for the tuple $(\clid,X,C,x,y)$
is equal to 
\added{$$\left(1-\prob(E_i)\right)^m=(1-p)^m\le e^{-mp}.$$}
\deleted{Thus, the}
\added{Thus, by the union bound, the} probability that there exists a tuple $(X,C,\clid,x,y)$ 
for which 
$E_i$ does not occur for all $i$
is smaller or equal to
\[|\mathcal{T}|(1-p)^m\le N^{k+l+3} e^{-mp} = e^{O((k+l)\log N) - mp}.\]
By choosing $m=\Theta((k+l)lk^2\log N)$ large enough,
the above \added{expression} gets strictly smaller than $1$ and therefore 
the probability of obtaining $(N,k,l)$-wcss is positive.
\added{Hence, by using the probabilistic method, such an object exists.}
\end{proof}

\subsection{Basic communication under SINR interference model}
\labell{s:basic:SINR}

Using introduced selectors,
 we provide some basic communication primitives on which we build 
our sparsification and clustering algorithms. As the proofs of stated
properties are fairly standard, we present them in Appendix.

Firstly, consider networks with constant density. Below, we state that
a fast efficient communication is possible in such a case.
\newcommand{\lconsdens}{(Sparse Network Lemma)
Let $\gamma\in\NAT$ be a fixed constant and $\eps\in(0,1)$ be the parameter
defining the communication graph.
There exists a schedule $\mathbf{L}_\gamma$ of length $O(\log N)$ such that each node $u$ 
transmits a message which can be received 
(at each point) in distance $\le 1-\eps$ from $u$
in an execution of $\mathbf{L}_\gamma$, 
provided there are at most $\gamma$ nodes in each unit ball.
}
\begin{lemma}\labell{l:cons:dens} 
\lconsdens
\end{lemma}
\noindent The schedule $\mathbf{L}_\gamma$ defined in Lemma~\ref{l:cons:dens} will be informally
called Sparse Network Schedule or shortly SNS.
%
\noindent The following lemmas imply that, for a successful transmission of a message on a link connecting a close pair,
it is sufficient that some set of constant size of close neighbors is not transmitting at the
same time (provided a round is free of some set of \added{conflicting} clusters of constant size). This fact will allow to apply wss (and wcss) for efficient communication
in the SINR model.
\newcommand{\lclose}{
There exists a constant $\kappa$ (which depends merely on the SINR parameters \added{and $\eps$}) which satisfies the following property. Let $u,v$ be a close pair of nodes in an unclustered set $A$. Then, there exists a set $A'\subseteq A$ such that $u,v\in A'$, $|A'|\le \kappa$ and $v$ receives a message transmitted
from $u$ provided $u$ is sending a message and no other element
from $A'$ is sending a message.
}
\begin{lemma}\labell{l:close}
\lclose
\end{lemma}
\newcommand{\lcloseclustered}{
Let $A$ be an $r$-clustered set for a fixed $r=O(1)$. Then, there exists 
constants $\kappa,\rho$ (depending only on $r,\eps$ and SINR parameters)
satisfying the following condition.
For each cluster $\clid$ and each close pair of nodes $u,v$
from $\clid$, there exists $A'\subseteq A$ of size $\le \kappa$ and
a set of clusters $C$ of size $\le \rho$ such that $u$ receives
a message transmitted by $v$ in a round satisfying the conditions:
\begin{itemize}
\gora{4}
\item
no node from $A'$ (except of $v$) transmits a message in the round,
\gora{4}
\item
the round is free of clusters from $C$.
\end{itemize}
}
\begin{lemma}\labell{l:close:clustered}
\lcloseclustered
\end{lemma}

\subsection{Proximity graphs}\labell{ss:prox:graph}
The idea behind sparsification algorithm extensively applied in our solutions is to
repeat several times the following procedure:
identify a graph \added{including as edges} all close pairs and ``sparsify this''
graph (switch off some nodes) appropriately. To this aim, we introduce the notion
of \emph{proximity graph}.
For a given (clustered) set of nodes $X$, a \emph{proximity graph} $H(X)$ of $X$ is any graph on this set 
such that: 
\begin{enumerate}[(i)]
\gora{4}
\item
vertices of each close pair $u,v$ are connected by an edge,
\gora{4}
\item
the degree of each node of the graph is bounded by
a fixed constant,
\gora{4}
\item
$\cluster(u)=\cluster(v)$ for each edge $(u,v)$ of $H(X)$.
\end{enumerate}
Using well-known strongly selective families (ssf) and Lemma~\ref{l:close}, one can build a schedule $S$
of length $O(\log N)$ such that each close pair of an unclustered network
exchange messages during an execution of $S$. 
However, this property is not sufficient for fast construction of a proximity graph,
since nodes may also
receive messages from distant neighbors during an execution of $S$ which migh
result in large degrees.
Moreover, each node $v$ knows
only received messages after an execution of $S$, but it is not aware
of the fact which of its messages were received by other nodes. 
Finally, a direct application of ssf for clustered network/set requires
additional increase of time complexity.
In order to build proximity graphs efficiently, we use the algorithm
ProximityGraphConstruction (Alg.~\ref{alg:prox:graph}). This algorithm relies on
properties of witnessed (cluster aware) strong selectors.

Our construction \added{(Alg.~\ref{alg:prox:graph})} builds on the following observations. Firstly, if $u$ can hear
$v$ 
in a round in which $w$ is transmitting as well,
then $u, w$ is for sure not a close pair (otherwise, $w$ generates interferences which prevents reception of the message from $v$ by $u$). Secondly, 
given a close pair $(u,v)$, $u$ can hear $v$ in a round in which $v$ transmits, provided:
\begin{itemize}
\gora{4}
\item unclustered case: none of the other $\kappa$ closest to $u$ nodes transmits (see Lemma~\ref{l:close}),
\gora{4}
\item clustered case: none of the other $\kappa$ closest to $u$ nodes from $\cluster(u)$ transmits,
and no node from $\rho$ ``conflicting'' clusters transmit (see Lemma~\ref{l:close:clustered}),
\end{itemize}
where $\kappa, \rho$ are the constants from Lemma~\ref{l:close} and Lemma~\ref{l:close:clustered}.

Given an $(N,\kappa)$-wss/$(N,\kappa,\rho)$-wcss $\mathbf{S}$ for constants $\kappa, \rho$ from Lemmas~\ref{l:close} and \ref{l:close:clustered}, one can build a proximity graph in $O(\log N)$ rounds using the following distributed algorithm at a node $v$
(see pseudocode in Alg.~\ref{alg:prox:graph} and an illustration on Fig.~\ref{fig:prox}):
\begin{itemize}
\gora{4}
\item
Exchange Phase: Execute $\mathbf{S}$.
\vspace*{-3pt}
\item
Filtering Phase:
 \begin{itemize}
 \item
Determine the set $C_v$ of all nodes $u$ such that $v$ has received a message from $u$ during $\mathbf{S}$ and $v$ has not received any other message in rounds in which $u$ is transmitting (according to $\mathbf{S}$). 
\textit{Remark:} ignore messages from other clusters in the clustered case.
 \item
If $|C_v|>\kappa$, then remove all elements from $C_v$.
 \end{itemize}
\vspace*{-3pt}
\item
Confirmation Phase
 \begin{itemize}
\vspace*{-3pt}
\item
Send information about the content of $C_v$ to other nodes in consecutive $|C_v|$ repetitions
of $\mathbf{S}$.
\vspace*{-3pt}
\item
Choose $E_v=\{w\in C_v\,|\, v\in C_w\}$ as the set of neighbours of $v$ 
in the final graph. (That is, $v$ exchange messages with its neighbours during an execution
$\mathbf{S}$.)
 \end{itemize}
\end{itemize}


\comment{
We show in Lemma~\ref{l:monograph} more formally that this algorithm actually
builds a proximity graph.
}
\newcommand{\algprox}{
\ifshort
\begin{algorithm}[h]
\else
\begin{algorithm}[H]
\fi
	\caption{ProximityGraphConstruction($v$)}
	\label{alg:prox:graph}
	\begin{algorithmic}[1]
		\State $E_v \leftarrow \emptyset$
		\State $\mathbf{S} \leftarrow $  $(N,\kappa)$-wss or $(N,\kappa,\rho)$-wcss,  common to all nodes.\Comment{$\kappa,\rho$ -- constants from L.~\ref{l:close} or \ref{l:close:clustered}}
		\State Execute $\mathbf{S}$ with message containing ID of $v$ and $\cluster(v)$.\Comment{{\bf Exchange Phase.} }
		\State Let $U_v$ be the set nodes which successfully delivered a message to $v$ during Exchange Phase.
		\State $C_v \leftarrow U_v$\Comment{{\bf Filtering Phase.} }
		\For{{\bf each} $u,w\in U_v$}
		\If{in some round $u$ and $w$ transmitted {\bf and} $v$ heard $u$}\Comment{lookup in the schedule $\mathbf{S}$}
		\State $C_v \leftarrow C_v \setminus \{w\}$
		\EndIf
		\EndFor
		\If{$|C_v|>\kappa$}
					\State $C_v\leftarrow \emptyset$
				\EndIf
		\For{{\bf each }$u \in C_v$}\Comment{{\bf Confirmation Phase} }
		\State Execute $\mathbf{S}$ with $\langle v, u\rangle$ as the message.
		\EndFor
		\For{{\bf each} message $\langle w, v\rangle$ received during the confirmation phase}
		\If{$w\in C_v$}
		\State $E_v \leftarrow E_v \cup \{w\}$
		\EndIf
		\EndFor
	\end{algorithmic}
\end{algorithm}
}
\iffull
\algprox
\fi

\begin{lemma} (Close Neighbors Lemma)
	\labell{l:closest}
ProximityGraphConstruction executed on a (clustered) set $X$ builds a proximity graph $H(X)$ of constant degree in $O(\log N)$ rounds. 
Moreover, ProximityGraphConstruction builds a schedule $\mathbf{S}$
of size $O(\log N)$ such that $u$ and $v$ exchange messages during $\mathbf{S}$ for each edge 
$(u,v)$ of $H(X)$.
\end{lemma}

\newcommand{\prooflclosest}{
\begin{proof}
First, we 
show that each close pair in $X$ is connected by an edge in $H$, i.e.,
$u\in E_w$ and $w\in E_u$ after an execution of ProximityGraphConstruction on $X$.
		
		By Lemma~\ref{l:close} we know that, if $u$ is the only transmitter among $\kappa$ nodes closest to $w$ (including $w$), then $w$ receives the message. By the 
		fact that the schedule	$\mathbf{S}$ is a $(N,\kappa,\rho)$-wcss, we know that there is a round satisfying this condition during an execution of $\mathbf{S}$. 
More precisely, we choose $\kappa$ and $\rho$ from Lemma~\ref{l:close:clustered}.
Thus, $u\in U_w$ and $w\in U_u$ after Exchange phase 
(see Fig.~\ref{fig:prox}(a)). 
		
Since $u,w$ is a close pair, it is
impossible that $u$ receives a message from $x\neq w$ in a round in which $w$ transmits 
a message (since $\beta>1$ and $d(u,x)>d(u,w)$). Thus, $w$ ($u$, resp.) belongs to $C_u$ ($C_w$, resp.) after Filtering Phase.
    
A node can also purge the candidate set $C_v$ during Filtering Phase, if the set of candidates contains more than $\kappa$ nodes. However, as $\mathbf{S}$ is $(N,\kappa,\rho)$-wcss, if $u,w\in X$ is a close pair then the sizes of $C_u$ and $C_w$ ar at most $\kappa$. Indeed, let $U$ denote the set of $\kappa$ nodes closest to $u$ (including $u$). Then, for any node $u_{\text{far}}\not\in U$ which is not among $\kappa$ closest nodes to $u$, there is a round in $\mathbf{S}$ in which $w$ transmits uniquely in $U$ and also $u_\text{far}\not\in U$ transmits by the witnessed strong selection property of $\mathbf{S}$. Thus, $u_\text{far}$ is eliminated from $C_u$ in Filtering Phase, as well as any other node not in $U$. So, the set of candidates $C_u$ for $u$ is a subset of $U$, thus its size is at most $\kappa$  (see Fig.~\ref{fig:prox}(b)).
    
    It is clear that the degree of resulting graph is at most $\kappa=O(1)$. The round complexity of the procedure is at most $(\kappa+1)|\mathbf{S}|=O(\log N)$.
\end{proof}
}
\iffull
\prooflclosest
\fi

\newcommand{\figelim}{
\begin{figure}[H]
	\centering
	\begin{subfigure}[t]{0.32\textwidth}
		\includegraphics[width=\textwidth]{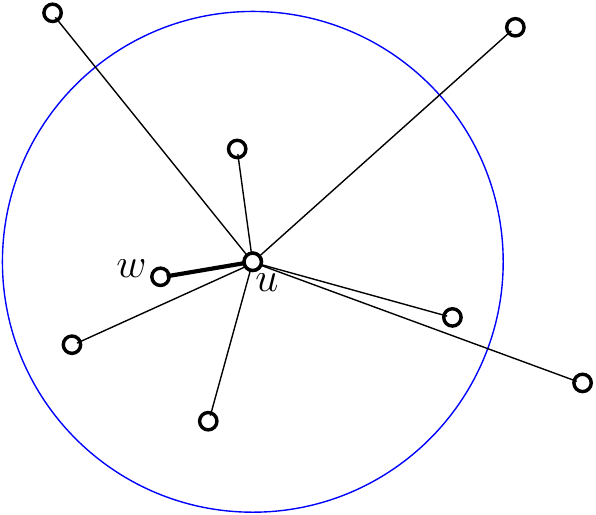}
		\caption{Exchange Phase: edges connect $u$ with elements of $U_u$. Some of them can be ``removed'' in Filtering Phase.}
	\end{subfigure}
	\hfill
	\begin{subfigure}[t]{0.32\textwidth}
		\includegraphics[width=\textwidth]{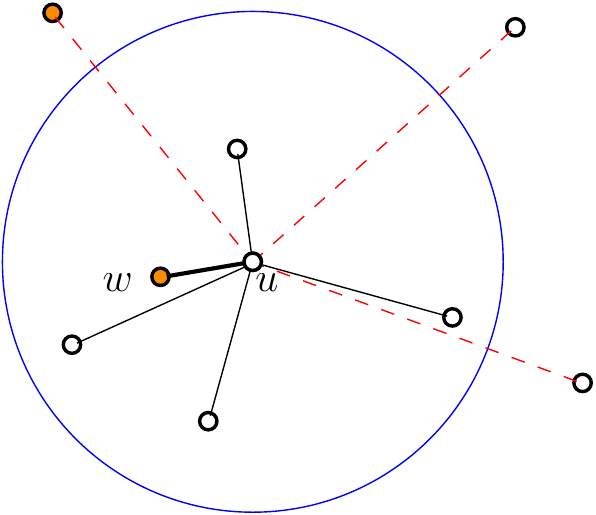}
		\caption{Filtering Phase: If $u$ hears some node $w$, it can ``remove'' all other nodes that transmitted in that round. The dashed edges denote
		nodes removed for sure.}
	\end{subfigure}
	\hfill
	\begin{subfigure}[t]{0.32\textwidth}
		\includegraphics[width=\textwidth]{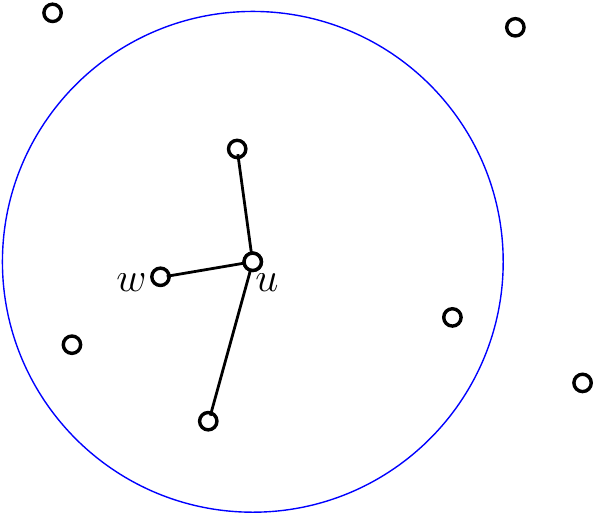}
		\caption{A possible result of the proximity graph construction. 
		Only close pairs are guaranteed to be connected by an edge.}
	\end{subfigure}
	\caption{An illustration of the proximity graph construction (Algorithm~\ref{alg:prox:graph}). The blue circle contains $\kappa$ closest nodes to $u$. 
	}
	\label{fig:prox}
\end{figure}
}
\ifshort
\figelim
\else
\figelim
\fi

%% file: sparsification.tex
In this section we describe a sparsification algorithm which
decreases density of a network by removing some nodes from dense
areas and assigning them to their ``parents'' which are not removed.
Simultaneously, the algorithm builds a schedule in which removed
nodes exchange messages with their parents.
Using sparsification as a tool, we develop other tools for (non)sparsified networks.
Finally, a clustering algorithm is given which partitions a network
(of awaken nodes) into clusters in time independent of the
diameter of a network.

\subsection{Sparsification}

A $c$-sparsification algorithm 
for a constant $c\le 1$ is a distributed ad hoc
algorithm $S$ which,
executed on a set of 
nodes $X$ of density $\Delta\le N$,
determines 
a set $Y\subseteq X$ such that:
\begin{enumerate}[a)]
\gora{4}
\item
density of $Y$ is at most $c\Delta$,
\gora{4}
\item
each $v\in X$ knows
whether $v\in Y$ and each $v\in X\setminus Y$ has assigned $\text{parent}(v)\in Y$ such
that $v$ and parent$(v)$ exchange messages during an execution of $S$ on $X$.
\end{enumerate}
Moreover, if $X$ is a clustered set, cluster$(v)$=cluster$(\text{parent}(v))$ for each
\deleted{$v$ with assigned $\text{parent}(v)$.}
\added{$v\in X\setminus Y$.}

Alg.~\ref{alg:sparsification} contains a pseudocode of our sparsification algorithm.
The algorithm builds a proximity graph of a network several times. Each time
a proximity graph $H$ is determined, an independent set $Y$ of $H$ is computed such
that:
\begin{itemize}
\item
for each dense cluster $\clid$, at least one element of $\clid$ is in $Y$ (clustered case) or 
\item
for each dense unit-ball $\cB$, there is an element in $Y$ located close to the center of $\cB$ \added{(unclustered case).}
\end{itemize}
Then, some elements of $Y$ are linked with their neighbours
in $Y$ by parent/child relation and removed from the set of nodes 
attending consecutive executions of ProximityGraphConstruction. 
\begin{figure}[H]
	\centering
	\begin{subfigure}[t]{0.45\textwidth}
  \includegraphics[width=1.0\textwidth]{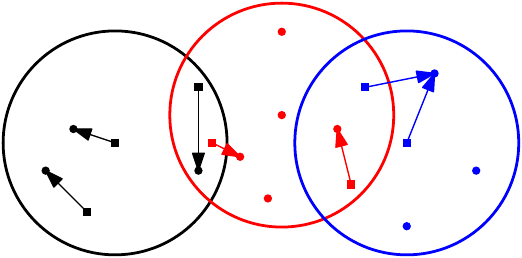}
		\caption{Clustered case. Child-parent relations is possible only between nodes of the same cluster.}
	\end{subfigure}
	\hspace*{20pt}
	\begin{subfigure}[t]{0.45\textwidth}
  \includegraphics[width=1.0\textwidth]{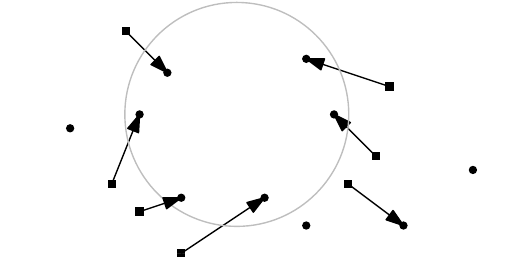}
		\caption{Unclustered case. An example demonstrates that the number of nodes in a dense unit ball (in gray) is not necessarily reduced, due to the fact that all nodes from the ball
		become parents of nodes outside of the ball.}
	\end{subfigure}
  \caption{An example of sparsification of a network. 
	Each edge corresponds to child-parent relation (the direction of an edge is from a child to its parent). The elements of the returned set of nodes are denoted by dots, the remaining nodes are denoted by squares.
	\label{f:spars}}
\end{figure}

In this way, the density of the set of nodes attending the algorithm
gradually decreases.
Details of implementation of the algorithm, including differences between clustered and unclustered variant, and analysis of its efficiency are presented below \added{(see also Fig.~\ref{f:spars}).}

\newcommand{\newchildren}{\text{NewChl}}
\newcommand{\globalchildren}{\text{GlobChl}}
\newcommand{\parents}{\text{Prnts}}
\newcommand{\actnodes}{\text{Active}}

\begin{algorithm}[h]
	\caption{Sparsification$(\Gamma,X)$}
	\label{alg:sparsification}
	\begin{algorithmic}[1]
	\State For each $v\in X$: parent$(v)=\bot$, children$(v)=\emptyset$
	\State $\actnodes\gets X$, $\parents\gets\emptyset$, $\globalchildren\gets\emptyset$
	\For{$i=1,2,\ldots,\Gamma$}
		\State $H\gets$ ProximityGraphConstruction$(\actnodes)$
		\State \label{l:ind:set} $Y\gets \text{IndependentSet}(H)$\label{line:is}\Comment{Implementions different for clustered/unclustered $X$}
		\State $\newchildren\gets\{v\in \actnodes\,|\, v\not\in Y\mbox{ and }N_v^H\cap Y\neq\emptyset\}$\Comment{$N_v^H$ denotes the set of neighbors of $v$ in $H$}
		\For{each $v\in \newchildren$}
			\State parent$(v)\gets\min(N_v^H\cap Y)$
			\State $\text{children}(\text{parent}(v))\gets \text{children}(\text{parent}(v))\cup \{v\}$
		\EndFor
		\State $\parents\gets \parents\cup \{v\in \actnodes\,|\, \text{children}(v)\neq \emptyset\}$
		\State $\globalchildren\gets \globalchildren\cup \newchildren$ 
		\State $\actnodes\gets \actnodes\setminus (\parents\cup \globalchildren)$
	\EndFor
	\State return $\actnodes\cup \parents$
	\end{algorithmic}
\end{algorithm}

For a clustered network, IndependentSet$(G)$ is 
chosen as the set of local minima in a proximity graph, that is,
$Y\gets \{v\,|\, \text{ID}(v)<\min(\{\text{ID}(u)\,|\, u\in N_v^H\})\}.$
%

\begin{lemma}\labell{l:sparsification:clustered}
Alg.~\ref{alg:sparsification} is a $\frac34$-sparsificaton algorithm for clustered networks, it works
in time $O(\Gamma \log N)$.
\end{lemma}
\begin{proof}
First, let us discuss an implementation of the algorithm
in a distributed ad hoc network. A proximity graph $H$
is built by Alg.~\ref{alg:prox:graph} in $O(\log N)$ rounds.
Moreover, after an execution of Alg.~\ref{alg:prox:graph}, all nodes
know their transmission patterns in the schedule $S$
of length $O(\log N)$ such that each pair of neighbors
$u,v$ in $G$ exchange messages during $S$. Therefore,
each node $v$ knows its neighbors $N_v^H$ in $H$ as well.
%
In order to execute the remaining steps in the for-loop, nodes determine locally whether they belong to $\newchildren$.
Then elements of $\newchildren$ choose their parents and send messages
to them using the schedule $S$. Simultaneously, each node updates the local children$(v)$ variable if it receives messages from its new child(ren).
(Note that a directed graph connecting children with their parents is acyclic.

Lemma~\ref{l:closest} implies that the graph $H$ built by
ProximityGraphConstruction connects by an edge each
close pair. Thus, by Lemma~\ref{l:density:close}.2,
$H$ contains
at least one edge connecting elements of $\clid$, for each
dense cluster $\clid$.
Thus, $Y$ contains at least one element of $\clid$
and, since edges connect only nodes of the same cluster,
at least one element of $\clid$ determines its parent inside $\clid$.
As the result, at least two elements of $\clid$ are removed
from $\actnodes$ in each execution of ProximityGraphConstruction.

These properties guarantee that each repetition of the 
main for-loop decreases the number of elements of $\actnodes$ in each
dense cluster. 
Therefore, 
each cluster
contains at most $\Gamma/2$ elements of $\actnodes$
after $O(\Gamma)$ repetitions of the loop.

Let $\actnodes'$, $\parents'$, $\globalchildren'$ be the intersections of $\actnodes$, $\parents$ and $\globalchildren$
with a cluster $\clid$.
Above, we have shown that $|\actnodes'|\le \Gamma/2$, while the algorithm
returns $\actnodes'\cup \parents'$. Thus, in order to prove the lemma,
it is sufficient to show that $|\parents'|\le \Gamma/4$.
For each cluster $\clid$, each element of $\parents'$ from this cluster has assigned nonempty subset of the cluster as children of $v$ and the sets of children of nodes are disjoint. Thus, the number of elements of $\parents'$ in the cluster is at most as large
as the number of elements of $\globalchildren'$ in that cluster: $|\parents'|\le |\globalchildren'|$.
Summing up, $|\actnodes'\cup \parents'\cup \globalchildren'|\le \Gamma$, $|\actnodes'|\le \Gamma/2$, $|\parents'|\le |\globalchildren'|$.
As $\actnodes'$, $\parents'$, and $\globalchildren'$ are disjoint, the finally returned set $\actnodes'\cup \parents'$ contains at most
$\frac34\Gamma$ elements from the considered cluster $\clid$.
\end{proof}

For unclustered networks IndependentSet$(G)$ is 
determined by a simulation of the distributed Maximal Independent
Set (MIS) algorithm for the LOCAL model \cite{SchneiderW08} which, for graphs
of constant degree, works in time $O(\log^*n)$ \added{and sends $O(\log n)$-size messages}. 
As ProximityGraphConstruction builds an $O(\log N)$ time schedule in which
each pair of neighbors $u,v$ (of a built proximity graph) exchange messages, we can simulate each step
of the
algorithm from \cite{SchneiderW08} in $O(\log N)$ rounds. 
In contrast to the clustered networks,
we are unable to guarantee that a \textbf{single} execution
of Sparsification reduces the density in an unclustered network. 
This is due to the fact that a node from dense unit ball might
become a parent of a node outside of this ball \added{(see Fig.~\ref{f:spars}(b));} as a result
the number of elements in the considered unit ball is not reduced
during an execution of Sparsification.
Therefore, we have to deal with unclustered case more carefully.
Let $l\gets \chi(5,1-\eps)$ and let SparsificationU$(\Gamma,X)$ be an  
algorithm which executes Sparsification$(\Gamma,X_i)$ for $i=0,\ldots, l-1$,
where $X_0=X$, $X_i$ is the set returned by Sparsification$(\Gamma,X_{i-1})$ and $X_l$ is the set returned as the final result.
%
\begin{algorithm}[h]
	\caption{SparsificationU$(\Gamma,X)$}
	\label{alg:sparsification:u}
	\begin{algorithmic}[1]
	\State $X_0\gets X$
	\State $l\gets \chi(5,1-\eps)$
	\For{$i=0,\ldots, l-1$}
		\State $X_{i+1}\gets$ Sparsification$(\Gamma,X_i)$
		\State $S_i$: the schedule of Sparsification$(\Gamma,X_i)$
	\EndFor
	\State return $(X_1,\ldots,X_l)$, $(S_0,\ldots,S_{l-1})$
	\end{algorithmic}
\end{algorithm}

\newcommand{\lsparsuncluster}{
An execution of SparsificationU$(\Gamma,U)$ on an unclustered set $U$ of density $\Gamma$ returns a sequence of set $X_0\supseteq X_1\supseteq\cdots\supseteq X_l$ and schedules $S_0,\ldots,S_{l-1}$ such that 
the density of $X_l$ is at most $\frac34\Gamma$, each node $x\in X_i\setminus X_{i+1}$
exchange messages with parent$(x)\in X_{i+1}$ during $S_i$.
Morover, SparsificationU$(\Gamma,U)$ works in time $O(\Gamma \log N\log^* N)$.
}
\begin{lemma}\labell{l:sparsification:unclustered}
\lsparsuncluster
\end{lemma}
\newcommand{\plsparsuncluster}{
\begin{proof}
\added{The analysis of the sparsification 
algorithm for clustered networks relies on the 
fact that there is a close pair in each dense cluster
$\clid$ and therefore the number of active elements in
$\clid$ is reduced after each iteration of the for-loop
in Alg.~\ref{alg:sparsification}.}
%
%
In the unclustered network, it might be the case that 
\deleted{nodes from a dense unit-ball $\cB$ become parents of other nodes
which do not belong to $\cB$. }
\added{there is no close pair in a dense unit-ball.}
That is why the reasoning from
Lemma~\ref{l:sparsification:clustered} does not apply here.

We define an auxiliary notion of \emph{saturation}.
For a unit ball $\cB=B(x,1)$, the saturation of $\cB$ with respect
to the set of nodes $X$ is the number of elements of $X$ in
$B(x,5)$. As long as $\cB$ contains at least $\Gamma/2$ nodes
in an execution of ProximityGraphConstruction, there is a close pair $u,v$ in $B(x,5)$ (Lemma~\ref{l:density:close})
and therefore $u,v$ are connected by an edge in $H$. Thus, $u$
or $v$ is not in the computed MIS (line~\ref{l:ind:set}); wlog assume that $u$ is not in MIS.
Thus, $u$ is dominated by a node from MIS and therefore it becomes
a member of $\globalchildren$ and it is switched off. This in turn decreases saturation
of $B(x,5)$. Thus, an execution of Sparsification results either in
decreasing the number of nodes in $\cB$ to $\le \Gamma/2$ or in
reducing saturation of $\cB$ by at least $\Gamma$. As \added{$B(x,5)$}
might be covered by $l=\chi(5,1)$ unit balls, there are at
most $l\Gamma$ nodes in $B(x,5)$ and saturation can be reduced
at most $l\Gamma$ times.
Hence, $l$ repetitions of Sparsification eventually leads to the reduction
of the number of 
\deleted{elements of $\cB$ to $\le \Gamma/2$.}
\added{active elements located in $\cB$ to $\le \Gamma/2$.}
\end{proof}
} 
\ifshort
\plsparsuncluster
\else
\plsparsuncluster
\fi

\subsection{Full Sparsification}
A full sparsification algorithm is a distributed ad hoc
algorithm/schedule $S=S_1 S_2\ldots S_k$ which,
executed on an $r$-clustered set of nodes $A$ of density $\Gamma$
determines 
sets $A_0,A_1,\ldots,A_k$ such
that \added{$k=\log_{4/3}\Gamma$,} $A_{i+1}\subseteq A_{i}$ (each node $v\in A$ knows
whether $v\in A_i$), $A_0=A$ and, for each $i\in[k]$:
\begin{enumerate}[a)]
\gora{4}
\item
the density of $A_i$ is at most 
\added{$\max\{\Gamma(3/4)^i, \chi(r,1-\eps)\}$},
\comment{\gora{4}
\item
for each cluster $\clid$ present in  $A$, $A_k$ contains at least one element
of $\clid$,}
\gora{4}
\item
each $v\in A_i\setminus A_{i+1}$ has assigned $\text{parent}(v)\in A_{i+1}$ such
that $v$ and parent$(v)$ exchange messages during $S_i$ and
	cluster$(v)$=cluster$(\text{parent}(v))$. 
\end{enumerate}

\added{Alg.~\ref{alg:complete:sparsification} contains a pseudocode of our} full sparsification algorithm \added{(see also Fig.~\ref{fig:full:sparsification}).} In Lemma~\ref{l:csparsification:cluster}, the properties of this algorithm (following directly from Lemma~\ref{l:sparsification:clustered}) are summarized.
\begin{algorithm}[H]
	\caption{FullSparsification$(\Gamma,A)$}
	\label{alg:complete:sparsification}
	\begin{algorithmic}[1]
	\State $A_0\gets A$, 
	\State $X\gets A$, $\Lambda\gets\Gamma$, $k\gets\log_{\frac34}\Gamma$
	\For{$i=1,2,\ldots,k$}
		\State $Y\gets \text{Sparsification}(\Lambda,X)$
		\State $S_i\gets$ the schedule of $\text{Sparsification}(\Lambda,X)$
		\State $A_i\gets X\setminus Y$
		\State $X\gets Y$
		\State $\Lambda\gets \frac34 \Lambda$
	\EndFor
	\State return $(A_0,\ldots,A_k), (S_1,\ldots,S_k)$
	\end{algorithmic}
\end{algorithm}

\comment{
\begin{figure}[h]
	\centering
	\begin{subfigure}[t]{0.32\textwidth}
		\includegraphics[width=\textwidth]{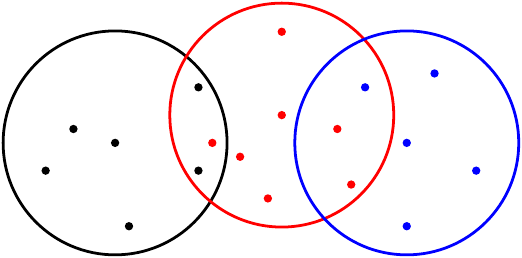}
		\caption{The clustered set $A$, where colors correspond to clusters.}
	\end{subfigure}
	\hfill
	\begin{subfigure}[t]{0.32\textwidth}
		\includegraphics[width=\textwidth]{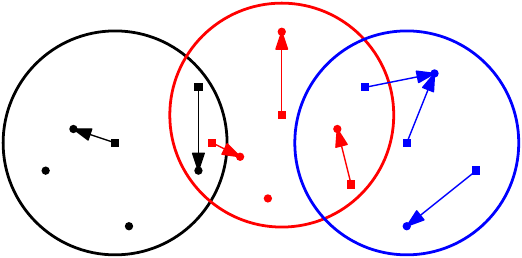}
		\caption{The result of the first execution of Sparsification. 
		Nodes from $A_1$ are denoted by squares.}
	\end{subfigure}
	\hfill
	\begin{subfigure}[t]{0.32\textwidth}
		\includegraphics[width=\textwidth]{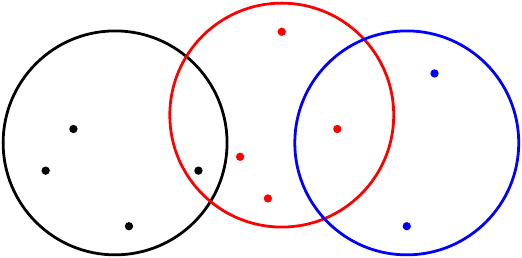}
		\caption{The result of the second execution of Sparsification. Nodes from $A_2$ are denoted by empty squares.}
	\end{subfigure}
	\caption{An illustration of full sparsification. Each edge correspond to child-parent relation (the direction of an edge is from a child to its parent).
	}
	\label{fig:full:sparsification}
\end{figure}
}
\begin{figure}[h]
	\centering
	\begin{subfigure}[t]{0.45\textwidth}
		\includegraphics[width=1.0\textwidth]{sparse1.pdf}
		\caption{The clustered set $A$, where colors correspond to clusters.}
	\end{subfigure}
	\hfill
	\begin{subfigure}[t]{0.45\textwidth}
		\includegraphics[width=1.0\textwidth]{sparse2.pdf}
		\caption{The result of the first execution of Sparsification. 
		Nodes from $A_1$ are denoted by squares.}
	\end{subfigure}
	\\
	\begin{subfigure}[t]{0.45\textwidth}
		\includegraphics[width=1.0\textwidth]{sparse3.pdf}
		\caption{The set on which Sparsification is executed for the second time (dots from Figure (b)).}
		\end{subfigure}
			\hfill
	\begin{subfigure}[t]{0.45\textwidth}
		\includegraphics[width=1.0\textwidth]{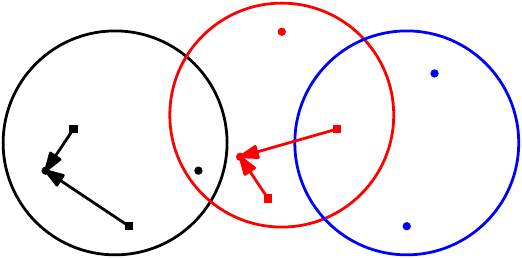}
		\caption{The result of the second execution of Sparsification. Nodes from $A_2$ are denoted by  squares.}
	\end{subfigure}
	\caption{An illustration of full sparsification. Each edge correspond to child-parent relation (the direction of an edge is from a child to its parent).
	}
	\label{fig:full:sparsification}
\end{figure}
\begin{lemma}\labell{l:csparsification:cluster}
Algorithm~\ref{alg:complete:sparsification} is a full sparsificaton algorithm for clustered networks which works
in time $O(\Gamma \log N)$.
\end{lemma} 

\subsection{Imperfect labelings of clusters}
Using $r$-clustering of a set $X$ with density $\Gamma$, it is possible
to build efficiently an imperfect $c$-labeling of $X$, where 
$c$ is a constant which 
depends merely on $r$ and SINR parameters.
\begin{lemma}\labell{l:labeling}
Assume that an $r$-clustering of a set $X$ of density $\Gamma$ is given.
Then, it is possible to build $c$-imperfect labeling of $X$ in
$O(\Gamma\log N)$ rounds, where $c$ depends merely on $r$ and SINR
parameters.
\end{lemma}
\begin{proof}
Let $S=S_1,\ldots,S_k$ and $A_1,\ldots,A_k$ be the schedules and the sets
obtained as the result of an execution of FullSparsification$(\Gamma,X)$.
FullSparsification$(\Gamma,X)$ splits each cluster in $O(1)$ trees, defined by the child-parent relation build during an execution of FullSparsification. 
(Indeed, the elements of $A_k$ are the roots of the trees.)
%
The schedule $S$ (and its reverse $S^R$) of length $O(\Gamma\log N)$ 
allows  for bottom-up (top-down) communication inside these trees. 
Using $S$ and $S^R$ one can implement tree-labeling algorithm as follows.
First, each node learns the size of its subtree in a bottom-up
communication. Then, the root starts top-down phase, where each node assigns the smallest label in a given range to itself and assigns appropriate subranges to its children. 
More precisely, the root starts from the range
$[1,m]$, where $m$ is the size of the tree. Given the interval $[a,b]$, each node assigns $a$ as its own ID
and splits $[a+1,b]$ into its subtrees.
\end{proof}

\subsection{Reduction of radius of clusters}
In this section we show how to reduce the radius of a clustering.
Given an $r$-clustering of a set $X$ of density $\Gamma$, our goal is to
build an $1$-clustering of $X$. The idea is to repeat the following
steps several times. First, $X$ is \added{fully} sparsified,
\added{i.e.,} $O(1)$ nodes remain
from each non-empty cluster. Then, a minimum independent set (MIS) of this sparse set
is determined on the graph with edges $(u,v)$ connecting $u$, $v$ which
exchange messages during an execution of Sparse Network Schedule
(see Lemma~\ref{l:cons:dens}). The elements of MIS become the centers of new clusters in the new
$1$-clustering
and they execute Sparse Network Schedule.
A node $v$ (not in MIS) which receives a message from $w$ during this
execution of SNS, becomes an element of $\cluster(w)$ and
\added{it} is removed from further consideration.
As we show below, a $1$-clustering of $X$ can be obtained in this way efficiently by Alg.~\ref{alg:radius:reduction}.
\begin{algorithm}[h]
	\caption{RadiusReduction$(\Gamma,X,r)$}
	\label{alg:radius:reduction}
	\begin{algorithmic}[1]
	\State For each $v\in X$: newcluster$(v)\gets \bot$ \Comment{$\bot$ means ``undetermined''}
	\For{$i=1,2,\ldots,\chi(r+1,1-\eps)$}
		\State $(X_0,\ldots,X_k), (S_0,\ldots,S_k)\gets$ FullSparsification$(\Gamma,X)$
		\State Execute Sparse Network Schedule $L_\gamma$ from Lemma~\ref{l:cons:dens} on $X_k$ \label{s:rr5} 
		\State Let $G(Y,E)$, where $E=\{(u,v)\,|\, u,v\text{ exchange messages during } L_\gamma\}$
		and 
		$Y=\{v\,|\, \exists_u\  (u,v)\in E\}$
		\State \label{s:rr6} $D\gets \text{MIS}(G)$ \Comment{simulation of maximal independent set alg.~from \cite{SchneiderW08}}
		\State \label{s:rr7} 
		Execute Sparse Network Schedule $L_\gamma$ from Lemma~\ref{l:cons:dens} on $D$
		\Comment{Local Broadcast from $D$, using L.~\ref{l:cons:dens}}
		\For{each $v\in X\setminus D$}
			\If{$v$ received a message from $u\in D$ during execution of $L_\gamma$ in line~\ref{s:rr7}} \label{s:rr9}  
			  \State 
			$\text{newcluster}(v)\gets u$\Comment{Choose arbitrary $u$ if $v$ received many messages}	\label{s:rr10}	
			\EndIf
		\EndFor
		\State $X\gets X\setminus (D\cup \{v\,|\, \text{newcluster}(v)\neq \bot\})$
	\EndFor
	\end{algorithmic}
\end{algorithm}
\begin{lemma}\labell{l:clustering:reduction}
Assume that a $r$-clustering of a set $X$ of density $\Gamma$ is given for a fixed constant $r\ge 1$.
Then, Algorithm~\ref{alg:radius:reduction} builds $1$-clustering of $X$ in $O((\Gamma+\log^* N)\log N)$ rounds.
\end{lemma}
\begin{proof}
By Lemma~\ref{l:sparsification:clustered}, the set 
$X_k$
has density $c'=O(1)$ and it contains at least one element from each nonempty
cluster of $X$. 
Each pair of nodes from \added{$X_k$} in distance $\le 1-\eps$
exchange messages during step~\ref{s:rr5}, therefore the graph $G$ contains the communication graph
of \added{$X_k$}. As $D$ computed in step~\ref{s:rr6} is a MIS of \added{$X_k$} in
the graph $G$ and \added{$X_k$} contains at least one element from
each nonempty cluster, there is an element of $D$ in close neighborhood of each cluster. 
More precisely, for a dense cluster $\clid$ with nodes located inside $B(x,r)$, there
is an element of the computed maximal independent set (MIS) in $B(x,r+1)$. 
Indeed, by Lemma~\ref{l:csparsification:cluster}, there is $y$ from $\clid$ in $Y$.
Thus, $y\in B(x,r)$, where $x$ is the center of the cluster $\clid$. Either $y$ is in
the computed MIS or $y$ is in distance $\le 1$ from an element $z$ from the MIS.
In the \replaced{latter}{former} case, $z\in B(x,r+1)$.
Then, steps \ref{s:rr7}--\ref{s:rr10} assign all nodes in distance $\le 1-\eps$ (and some
in distance $\le 1$) from elements of $D$ 
\added{(including the above defined $y$ or $z$)}
to new clusters included in unit balls with centers at the elements of $D$. Thus, the nodes from each ``old'' cluster  are assigned
to new clusters after $\chi(r+1,1-\eps)$ repetitions of the main for-loop.
\end{proof}


\subsection{Clustering algorithm}
In this section we provide an algorithm which, given
an unclustered set $A$, builds a $1$-clustering
of $A$. 
The algorithm consists of two main parts.
In the former part, the sequence of sets $A_0\supseteq\cdots\supseteq A_{kl}$
is built using SparsificationU, for $k=\log_{4/3}\Gamma$ and
$l=\chi(5,1-\eps)$.
By Lemma~\ref{l:sparsification:unclustered}, the density of
$A_{il}$ is at most $\Gamma(\frac34)^i$. Moreover, a sequence of schedules
$S_0,\ldots,S_{kl}$ is built such that each $v\in A_i\setminus A_{i+1}$ exchange
messages with parent$(v)\in A_{i+1}$.
In the latter part, we start from $1$-clustering of $A_{kl}$, which is obtained
by assigning each node $v\in A_{kl}$ to a separate cluster.
Then, given an $1$-clustering of $A_i$ for $i>0$,
we get $2$-clustering of $A_{i-1}$ by executing $S_{i-1}$ with messages
equal to cluster IDs of transmitting nodes. The elements of $A_{i-1}$ choose
clusters of their parents. Using RadiusReduction (Lemma~\ref{l:clustering:reduction}),
the obtained $2$-clustering is transformed into an $1$-clustering.
A pseudocode is presented in Alg.~\ref{alg:clustering}.

\begin{algorithm}[h]
	\caption{Clustering$(\Gamma,A)$}
	\label{alg:clustering}
	\begin{algorithmic}[1]
	\State $k\gets \log_{4/3}\Gamma$
	\State $X\gets A$, $A_0\gets A$
	\State $\Lambda\gets\Gamma$, $k\gets\log_{\frac34}\Gamma$, $l\gets \chi(5,1-\eps)$
	\For{$i=1,2,\ldots,k$}
		\State $(A_{(i-1)l+1},\ldots,A_{il}), (S_{(i-1)l+1},\ldots,S_{il})\gets \text{SparsificationU}(\Lambda,X)$
		\State $X\gets A_{il}$
		\State 
			$\Lambda\gets \frac34 \Lambda$
	\EndFor
	
	\State For each $v\in A_{kl}$: cluster$(v)\gets \text{ID}(v)$
	\State $X\gets A_{kl}$ 
	\State $\Lambda\gets 1$
	\For{$i=0,1,2,\ldots,kl$}
		\State Execute $S_{kl-i}$, 
			on $A_{kl-i}$, each $v\in A_{kl-i}$ sends cluster$(v)$ during the execution
		\State For each $v\in A_{kl-i-1}$: cluster$(v)\gets\text{cluster}(\text{parent}(v))$
		\State $X\gets X\cup A_{kl-i-1}$
		\State RadiusReduction$(\Lambda,X,2)$
		\State \textbf{if} $i\text{ mod } l=0$ \textbf{then} $\Lambda\gets \frac43\cdot\Lambda$
	\EndFor
	\end{algorithmic}
\end{algorithm}

\begin{theorem}\labell{t:clustering}
The algorithm Clustering builds $1$-clustering of an unclustered
set $A$ of density $\Gamma$ in time $O(\Gamma\log N\log^\star N)$.
\end{theorem}
\begin{proof}
Correctness of the algorithm follows from the properties of
SparsificationU and RadiusReduction (Lemmas~\ref{l:sparsification:unclustered} and \ref{l:clustering:reduction}).
The time of executions of SparsificationU is
$$O(\sum_{i=1}^k \Gamma\left(\frac34\right)^i\log N\log^* N)=O(\Gamma\log N\log^*N).$$
The time of executions of RadiusReduction is
$$O(\sum_{i=1}^k (\Gamma\left(\frac34\right)^i+\log^*N)\cdot\log N)=O(\Gamma\log N+\log^*N\log N).$$ Altogether, time complexity is $O(\Gamma\log N\log^\star N)$.
\end{proof}

%% file: commProblems.tex
In this section, we apply sparsification and clustering
in distributed algorithms for broadly studied communication problems, 
especially the local and global broadcast.

\subsection{Local broadcast}
Now, we can solve the local broadcast problem 
for a set $V$ of density $\Delta$ by the
following LocalBroadcast algorithm.
\begin{algorithm}[H]
	\caption{LocalBroadcast($V,\Delta$)}
	\label{alg:localBroadcast}
	\begin{algorithmic}[1]
		\State Clustering($\Delta,V$) \Comment{Theorem~\ref{t:clustering}}
		\State Imperfect labeling of $V$ \Comment{Lemma~\ref{l:labeling}}
		\For{$l=1,\ldots,\Delta$}
			\State Run SNS on nodes with label equal to $l$ \Comment{Lemma~\ref{l:cons:dens}}
		\EndFor
	\end{algorithmic}
\end{algorithm}
First, the clustering algorithm is applied which builds
a $1$-clustering of $V$ in time $O(\Delta\log N\log^\star N)$
(Theorem~\ref{t:clustering}). Then, an imperfect labeling of clusters can be formed
in time $O(\Delta\log N)$ (Lemma~\ref{l:labeling}).
Finally, an algorithm which executes $\Delta$ times 
Sparse Network Schedule (see Lemma~\ref{l:cons:dens}),
the $l$th execution of SNS is  performed by the nodes
with label $l$.
As each label appears $O(1)$ times in each cluster,
the density of the set of nodes with label $l$ (for each $l$) is $O(1)$ 
for each $l\in[N]$.
Thus, consecutive executions of SNS accomplish the local broadcast task
(see Lemma~\ref{l:cons:dens}).

\begin{theorem}\labell{t:local:broadcast}
The algorithm LocalBroadcast performs local
broadcast from a set $V$ of density $\Delta$ in time $O(\Delta\log N\log^\star N)$.
\end{theorem}

\comment{
In this section we present an algorithm for the local broadcast problem. The key idea is to select a sparse subset $\mathcal{L}$ of nodes that forms an independent set in the communication graph and run the broadcast process starting in each of the nodes of $\mathcal{L}$. Due to Proposition~\ref{p:stage2} we have that all nodes that are in an \emph{internal} state transmit a message to all the neighbors. After the broadcast is finished all nodes are in state $\mathcal{I}$. We exploit this property to perform local broadcast. However, the number of rounds needed for the broadcast process is dependent on the maximal distance from any node $v$ to the closest node $w\in\mathcal{L}$. Formally, let $d^*=\max_{v\in G}\min_{w\in \mathcal{L}} d(v,w)$, then it is sufficient to run $d^*+1$ phases of the broadcast algorithm. This yields a complexity of $O(d^*(\Delta+\log^* n)\log n)$. We construct the network sparsification $\mathcal{L}$ such that $d^*=O(\log n)$ and the complexity of our local broadcast algorithm is $O((\Delta+\log^* n)\log^2 n)$.

{\bf Network Sparsification.} 
We call a set $\mathcal{L}$ of nodes a \emph{network sparsification} if it satisfies the following criteria. The elements of the set are called \emph{leaders}.
\begin{itemize}
\item For each $u,w\in \mathcal{L}$ we have $d(u,w)>1-\eps$.
\item For any $v\not\in\mathcal{L}$ there is a leader $w\in \mathcal{L}$ such that $d(v,w)\le\log n$.
\end{itemize}

The idea to generate such sparsification is to use a BuildTree-like procedure (see Algorithm~\ref{alg:build_tree}) to organize nodes into trees and designate the roots to become the leaders. However, there is a major caveat to this: MonoGraphConstruction -- the key routine used in TreeMerging Algorithm assumes preexisting clustering of the network. Namely, each node $v$ has to be assigned to some cluster (indicated by the $\up(v)$ variable). In the broadcast algorithm it was realized by coupling nodes assigned to the same master node into one cluster. Unlike the broadcast scenario no nodes are assumed to be distinguished at the beginning of local broadcast. Since we do not have any initial clustering we redesign the procedure to work on unclustered networks at the cost of larger spread of trees  -- that is, we allow for a tree to occupy area of diameter $O(\log n)$ instead of $O(1)$.

The presented algorithm organizes a network in trees of depth at most $\log n$ such that for each two trees their roots are within distance at least $1-\eps$. 
The trees are built in a similar way to the communication trees (see Algorithm~\ref{alg:build_tree}). However, we adjust TreeMerging (Algorithm~\ref{alg:treemerging}) subroutine to take into the account the rank (the size of a tree rooted at the node), so that node with lower rank joins the node with higher rank. This makes the rank of the parent at least twice the rank of its children, thus bounds the depth of a tree by $O(\log n)$. 

\begin{algorithm}[H]
	\caption{NewTreeMerging$(\Gamma)$ in node $v$}
	\begin{algorithmic}[1]
		\State counter $\leftarrow 0$, $\up(v)\leftarrow 0$
		\While{parent$(v) = \bot$ {\bf and} counter $<\Gamma$}
		\State counter $\leftarrow$ counter$+1$
		\State $H\leftarrow$ ProximityGraphConstruction$(v)$
		\label{btm:monograph}
		\State Compute Maximal Independent Set of $H$
		\State If $v$ has no neighbours then wait until the next iteration of while loop.
		\If{$v \in MIS(H)$}
		\State Transmit message $\langle v, \text{MISNode} \rangle$
		\State Let $K=\{v_1,...,v_k\}$ be the set of nodes that sent handshake message to $v$.
		\State Choose any $w\in K$
		\If{rank$(w)>$ rank$(v)$}
		\State Send a parent declaration message $\langle w, v, \text{rank}(v) \rangle$
		\State parent$(v)\leftarrow w$
		\Else
		\State Send a message of acceptance $\langle v, w \rangle$.
		\State children$(v) \leftarrow $ children$(v) \cup \{w\}$
		\State rank$(v) \leftarrow $ rank$(v) + $ rank$(w)$
		
		\EndIf
		\Else
		\State Let $L$ be the set of received messages of type $\langle w, \text{MISNode}\rangle$. 
		\If{$L$ is non-empty} 
		\State Let $u$ be the node such that $\langle u, \text{MISNode}\rangle\in L$ having minimal ID.
		\State Transmit a handshake message $\langle u, v, \text{rank}(v) \rangle$. 
		\If{received an acceptance message from $u$}
		\State parent$(v)\leftarrow u$
		\ElsIf{received a parent declaration message from $u$}
		\State children$(v) \leftarrow $ children$(u) \cup \{u\}$
		\State rank$(v) \leftarrow $ rank$(v) + $ rank$(u)$
		\EndIf
		
		\EndIf
		\EndIf
		\EndWhile
	\end{algorithmic}
\end{algorithm}

The line of proof of correctness and complexity is the same as for TreeMerging Algorithm, however some details must be explained. The main idea was to show that in each iteration of the main while-loop, in every dense region, at least two trees merge. To carry this analysis we argue that in every dense region graph $H$ has at least one edge. This is true due to fact that in every such area there must be a locally minimal pair of nodes that hear each other (cf. Lemma~\ref{l:monograph}) during ProximityGraphConstruction. Then, at least one node near the edge is selected to $MIS(H)$ and coordinates the tree merging. This concludes the sketch of proof supplement to the correctness of the algorithm.

{\bf Local Broadcast via Broadcast.} Now, we have all the pieces to perform local broadcast. After running a version of BuildTree algorithm, adjusted as described in previous subsection, we obtain a network sparsification -- $\mathcal{L}$ in which we initiate the broadcast algorithm which terminates after $O((\Delta+\log^* n)\log^2)$ rounds. Thanks to its properties we know that all nodes had a chance to transmit a message to all its neighbors (see Proposition~\ref{p:stage2}), and thus complete the task of local broadcast.

\begin{theorem}
	Running the broadcast algorithm with a set of sources $\mathcal{L}$ solves local broadcast problem in $O((\Delta+\log^* n)\log^2)$ rounds.
\end{theorem}

}

\subsection{Sparse multiple-source broadcast}
In this section we consider 
{\em sparse multiple source broadcast} (SMSB) problem, a generalization
of the global broadcast problem. This generalization is introduced for
further applications for other communication
problems as  leader election
and wake-up.

The {\em sparse multiple source broadcast problem (SMSB problem)} is defined
as follows. At the beginning, the unique broadcast message is known to 
a set of distinguished nodes $S\subset V$ such that $d(u,v)>1-\eps$
for each $u,v\in S$, $u\neq v$. The problem is solved when
\begin{enumerate}
\item[(a)] the broadcast
message is delivered to all nodes of the network, AND
\item[(b)]
each node $v\in V$ transmitted a message in some round which was
received by all its neighbors in the communication graph.
\end{enumerate}

\noindent\textbf{The algorithm for SMSB problem.}
Let $V_i$ denote the set of nodes in the graph-distance $i$ from
$S$, i.e., $v\in V_i$ if a shortest path from an element of $S$
to $v$ has length $i$.
The algorithm works in phases.
In Phase~1,
Sparse Networks Schedule (Lemma~\ref{l:cons:dens}) is applied on the
set of distinguished nodes $S$. In this way all elements of $S$
transmit messages received by $V_1$, their neighbors in the communication
graph (and possibly some other nodes in geometric distance at most $1$). 
After receiving
the first message from $s\in S$, a node $v\in V\setminus S$
assigns itself to the cluster $s$. Hence, the
set of awaken nodes $L_1$ contains $V_1$ and we have 
$1$-clustering of $L_1$.

In Phase $i>1$, local broadcast is executed on $L_{i}$, the set of
nodes awaken\footnote{A node is awaken in the phase $j$ if it receives the broadcast message for the first time in that phase.} in Phase~$i-1$ . In this way, the set of nodes awaken in the first $i$
phases contains all nodes in graph distance $i$ from \added{$S$}, i.e.,
$\bigcup_{j\in[i]}V_j\subseteq \bigcup_{j\in[i]}L_j$. Moreover, 
we assure that all nodes awaken in a phase are $1$-clustered at the end of the phase. (Thus, each node knows its cluster ID in such clustering).

A phase consists of three stages. In Stage~1 of phase $i$, an \emph{imperfect labeling} of each cluster of $L_{i}$ is built. This means that each node $v$ is assigned a label $l_v$ such that, for each cluster, the number of nodes with the same label in the cluster is $O(1)$\added{(see Lemma~\ref{l:labeling}).} In Stage~2, Sparse Network Schedule (Lemma~\ref{l:cons:dens}) is executed $\Delta$ times. 
A node with label $l$ participates in the $l$th execution of SNS only. In this way, all nodes from $L_{i}$ transmit on distance $1-\eps$. All nodes
awaken in Stage~2 inherit cluster numbers from nodes which awaken them. In this
way, we obtain a $2$-clustering of $L_{i+1}$, the set of nodes awaken in Stage~2. The goal of Stage~3 is (given the obtained $2$-clustering) to build an $1$-clustering of awaken nodes.
\added{(See Fig.~\ref{fig:global} in Section~\ref{s:intro} for an example.)} 

\begin{algorithm}[H]
	\caption{SMSBroadcast($V,S$)}
	\label{alg:smsBroadcast}
	\begin{algorithmic}[1]
		\State SNS($S$) 
		\State $L_1\gets$ nodes awaken by elements of $S$, clustered by IDs of their neighbors from $S$
		\For{$i=1,2,\ldots$}
			\State Stage 1: imperfect labeling of $L_i$ \Comment{Lemma~\ref{l:labeling}}
			\State Stage 2: local broadcast from $L_i$, using assigned labels
			\Comment{Lemma~\ref{l:cons:dens}}
			\Statex \Comment{Awaken nodes form $L_{i+1}$ and inherit cluster IDs from nodes which awaken them.}
			\State Stage 3: RadiusReduction($\Delta,L_{i+1},2$)
			using
			$2$-clustering given by inherited cluster numbers 
			\Statex \Comment{Lemma~\ref{l:clustering:reduction}.}
		\EndFor
	\end{algorithmic}
\end{algorithm}
Using Lemmas~\ref{l:labeling}, \ref{l:cons:dens} and \ref{l:clustering:reduction}, we obtain the following result.
\begin{theorem}\labell{t:smsb:broadcast}
Algorithm SMSBroadcast solves sparse multiple source broadcast
problem in $O(D (\Delta+\log^* N) \log N)$ rounds. For $S$ of size $1$,
the algorithm solves the global broadcast problem.
\end{theorem}

\subsection{Other problems}
\newcommand{\otherproblems}{

\noindent\textbf{Wake-up problem}
In the wake-up problem \cite{JurdzinskiK16a}, some nodes become spontaneously \emph{active} at various
rounds and the goal is to activate the whole network. 
Nodes which do not become active spontaneously can be activated
by a message successfully delivered to them. 
Spontaneous wake-ups are controlled by an adversary, thus an algorithm
solving the problem should work for each pattern of spontaneous
activations. A node can not participate in an execution of an algorithm
as long as it is not active.
Time of an execution of a wake-up algorithm is the number of rounds between the 
spontaneous activation
of the first node and the moment when all nodes are activated.
We assume the model with \emph{global clock}, i.e., there is a central
counter of rounds and each node knows the current value of that counter.

For a while, assume that all spontaneous wake-ups appear at the same
round $r$. Then, starting at round $r$, we call Clustering (Alg.~\ref{alg:clustering})
on the set of spontaneously awaken nodes $S$ (see Lemma~\ref{l:sparsification:unclustered}). As a result, we get a nonempty
set $S'\subset S$ of constant density in $O(\Delta\log N\log^* N)$ rounds.
Then, SMSBroadcast$(V,S')$ activates all nodes on a network
in $O(D(\Delta+\log^* N)\log N)$ rounds.
Let $T(N,\Delta)=O(D(\Delta+\log^* N)\log N)$ be the exact number of rounds of this algorithm.
In order to adjust the above algorithm to arbitrary times of spontaneous
wake-ups, we start a separate execution of this algorithm at each round
$r$ whose number is divisible by $T(N,\Delta)$. In an execution starting
at round $r$, only nodes awaken before $r$ are considered to be activated
spontaneously.
Thus, we have the following result.

\begin{theorem}\labell{t:wakeup}
The wake-up problem in networks with global clock can be solved
in $O(D(\Delta+\log^* N)\log N)$ rounds.
\end{theorem}

\noindent\textbf{Leader election}
The \emph{leader election} problem is to choose (exactly) one node
in the whole network as the leader, assuming that nodes are awaken
in various rounds.
First, assume that all nodes start a leader election algorithm at the same round.
In order to elect the leader, we first execute 
Clustering (Alg.~\ref{alg:clustering})
on all nodes of a network which
takes $O(\Delta\log^* N\log N)$ rounds and determines the non-empty
set $S$ of constant density. Then, a unique member of $S$ is chosen
as the leader, in the following (standard) way. Observe that one can verify whether $S$ contains nodes 
with IDs in the range $[l,r]$ by executing SMSBroadcast$(V, S')$, where
$S'$ are nodes from $S$ with IDs in $[l,r]$. Thus, it is possible to
choose the leader from $S$ by binary search which starts from the
range $[1,N]$ and requires $O(\log N)$ executions of SMSBroadcast.
In order to adjust the above algorithm to arbitrary times of spontaneous
wake-ups of nodes, we start a separate execution of the algorithm at each round
$r$ whose number is divisible by $T(N,\Delta)$, where
$T(N,\Delta)=O(D(\Delta+\log^* N)\log^2 N)$ is the upper bound on the
time of the algorithm. In an execution starting
at round $r$, only nodes awaken before $r$ are considered to be activated
spontaneously. Finally, we have the following result.
\begin{theorem}\labell{t:leader}
The leader election problem in networks with global clock can be solved
in $O(D(\Delta+\log^* N)\log^2 N)$ rounds.
\end{theorem}

}
\ifshort
\comment{
By applying the clustering algorithm and SMSBroadcast, solutions for
other widely studied problems can be obtained, e.g., leader election
or wake-up. Examples of such solutions are presented in Appendix, Section~\ref{s:app:other}.
}
By applying the clustering algorithm and SMSBroadcast, solutions for
other widely studied problems can be obtained, e.g., leader election
or wake-up. 
\otherproblems
\else
\otherproblems
\fi

%% file: lowerbound.tex
In this section we prove a lower bound on the number of rounds needed to perform global broadcast. 
\ifshort
\else
\fi

\begin{theorem}\labell{t:lower:bound}
A deterministic algorithm for the global broadcast in the SINR network
works in time $\Omega(D\Delta^{1-1/\alpha}+\Delta)$, provided $N=\Omega(D\Delta)$
and the connectivity parameter $\eps>0$ is small enough.
\end{theorem}

In order to prove the above theorem, we build a family of networks called
\emph{gadgets}. Each network from the family consists of $\Delta+4$ nodes 
$s, v_0,\ldots,v_{\Delta+1}, t$
located on the line -- see Fig.~\ref{f:gadget} (in Section~\ref{s:lower:bound}) for locations and distances between nodes. Thus, in particular,
$s$ is connected by an edge in the communication graph with
$v_0,\ldots, v_{\Delta+1}$, $t$ is connected merely with $v_{\Delta+1}$ and only
messages from $v_{\Delta+1}$ can be received by $t$. One can show that, for each
deterministic algorithm, one can assign IDs to the nodes $s, v_0,\ldots,v_{\Delta+1}$
such that the broadcast message originating at $s$ will be delivered to $t$
after $\Omega(\Delta)$ rounds.
For the purpose of the lower bound, the main property of 
the locations of nodes in the gadget are:
\begin{enumerate}
\item[(a)]
The node $v_i$ does not receive any message, provided at least two nodes
from the set $\{v_j\,|\,j<i\}$ transmit at the same time. 
\item[(b)]
The node $t$ can receive a message from $v_{\Delta+1}$ ($v_{\Delta+1}$ is the only node from the gadget in distance $\le 1$ from $t$) only in the case that no other node from the gadget transmits at the same time.
\end{enumerate}
Using (a), we assign IDs to consecutive nodes $v_1, v_2, \ldots$ such that,
after $i$ rounds, the nodes $v_k$ for $k>2i$ do not have any information about
their location inside the set $\{v_l\,|\,l>2i\}$. Thus, the only available information
differentiating them are their IDs. As a result, we prevent a transmission of 
$v_{\Delta+1}$ in a round in which other nodes from the gadget do not transmit
for $\Omega(\Delta)$ rounds.

A natural idea to generalize such a result to $\Omega(D\Delta)$ is to connect
sequentially consecutive gadgets such that the target ($t$) of the $i$th gadget is identified with the source $s$ of the $(i+1)$st gadget. Such approach usually works
in the radio networks model, where no distant interferences appear. However, under the SINR constraints, the interference from distant nodes might potentially help, since it differentiate history of communication for nodes in various locations
(even the close ones). Therefore, the applicability of this idea under the SINR
constraints is limited.
Instead of identifying the target $t$ of the $i$th gadget with the source
of the $(i+1)$st gadget $s$, we separate each two consecutive gadgets by 
a long ``sparse'' path of nodes such that consecutive nodes are in distance $1-\eps$
(Fig.~\ref{f:siec}). 
In this way, we limit the impact
of nodes located outside of a gadget on communication inside the gadget.
Finally, our result is only $\Omega(D\Delta^{1-1/\alpha})$.
Below, we give a full formal proof following the above ideas.


\newcommand{\lowerboundproof}{	
First, we build a gadget (sub)network consisting of $\Delta+4$ nodes 
$s$, $v_0,\ldots,v_{\Delta+1}$, $t$ located 
on the line -- see Figure~\ref{f:gadget} and Figure~\ref{f:gadget-core} for locations of nodes of a gadget. 
%
The notion \emph{gadget} denotes actually a family of (sub)networks with location of nodes described on Figure~\ref{f:gadget} and Figure~\ref{f:gadget-core} and aribtrary IDs in the range $[N]$.
%
%
Observe that the diameter $D$ of the gadget network is equal to $2$ (indeed, $s$ is connected to $v_0,\ldots, v_{\Delta+1}$, $(v_i,v_j)$ is and edge for each $i\neq j$ and $t$ is connected with $v_{\Delta+1}$).
In Lemma~\ref{l:gadget} we show that any deterministic algorithm needs $\Omega(\Delta)$ rounds to deliver a message to $t$ (the \emph{target}), which originally is known only to $s$
(the \emph{source}). 
More precisely, we show that there exists an assignment of IDs for the nodes of the gadget such that delivery of a message from $s$ to $t$ takes $\Omega(\Delta)$ rounds. 
As our ultimate goal is to analyse networks which contain gadgets as
subnetworks, we make an additional assumption that a limited interference from outside of the
gadget might appear.

\begin{figure}[H]
	\centering
  \includegraphics[width=0.7\textwidth]{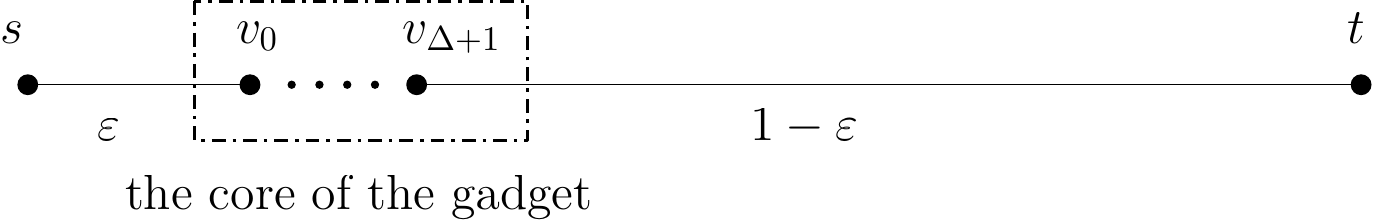}
  \caption{The gadget. The nodes $s$ and $t$ are called the source and the target of a given gadget. Importantly, $d(x,t)>1$ for each $x$ from the gadget except of $v_{\Delta+1}$. 
	\label{f:gadget}}
\end{figure}

\begin{figure}[H]
  \includegraphics[width=1.0\textwidth]{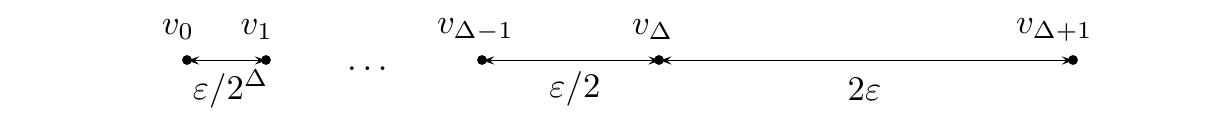}
  \caption{The core of the gadget, $d(v_i, v_{i+1}) = \eps/2^{\Delta-i}$ for $i<\Delta$. Observe, that $2\eps<d(v_0,v_{\Delta+1})<3\eps$. \label{f:gadget-core}}
\end{figure}


\begin{lemma}
\labell{l:gadget} For any deterministic algorithm $\mathcal{A}$ there exists a choice of identifiers for the nodes in the gadget such that the following holds. 
Let $\nu$ be the number satisfying 
$\frac{P/(4\eps)^\alpha}{\cN+\nu}=\beta$
and let $I\subseteq [N]$ be a set of {allowed} IDs
such that $|I|\ge \Delta+4$.
Assume that the interference coming from the outside of the core of the gadget $G$ 
in any node of the core 
is smaller than $\nu$ in every round. 
Moreover, assume that $s$ is the only awake node of the gadget at a given
starting round.
Then, it takes $\Omega(\Delta)$ rounds until the broadcast message 
is delivered to $t$ using algorithm $\mathcal{A}$.
\end{lemma}
\begin{proof}
The fact that the sequence of distances $d(v_0,v_1), d(v_1,v_2), \ldots$ forms a geometric
sequence implies the following property, provided $\eps>0$ is a small enough constant.
  \begin{fact}
  \labell{f:core}
    \begin{enumerate}
		\item
		If the nodes $v_i$ and $v_j$, for $i<j$ are transmitting 
		in a round then the nodes $v_{j+1},...,v_{\Delta+1}$ 
		do not receives a message in that round.
		\item
		If $v_{\Delta+1}$ is not the only transmitter from the set $\{v_0,\ldots,v_{\Delta+1}\}$ in a round, then $t$ does not receive a message in that round.
		\end{enumerate}
  \end{fact}
  
As $\{v_0,\ldots,v_{\Delta+1}\}$ are asleep at the beginning,
they are all awaken in a round in which $s$ sends its first message, thanks to the fact that interference from outside of the gadget's core is bounded by $\nu$.
In order to simplify notations, assume that the number
of the round in which $v_0,\ldots,v_{\Delta+1}$ receive 
the first message from $s$ is equal to $0$.
Recall that $v_{\Delta+1}$ is the only node 
from the gadget that can pass a message to the target $t$ (c.f.\ Figure~\ref{f:gadget} and Figure~\ref{f:gadget-core}). 
Our goal is to gradually assign IDs to $v_0, v_1, v_2,\ldots$ such that,
in each round $j\in[i/2]$, 
\begin{itemize}
\item
either no node from the set $\{v_0,\ldots,v_{\Delta+1}\}$ transmits,
\item
or at least two nodes from $\{v_0,\ldots,v_{i}\}$ transmit.
\end{itemize}
In this way the assignment of IDs prevents $v_{\Delta+1}$ from being the unique transmitter from $\{v_0,\ldots,v_{\Delta+1}\}$ for $\Omega(\Delta)$ rounds.
This fact in turn guarantees that $t$ does
not receive any message in $\Omega(\Delta)$ rounds.

Now, we describe the process of assigning IDs to 
$v_0,\ldots,v_{\Delta+1}$. 
Let $I_0=I$ of size $\ge \Delta+4$ be the set of allowed IDs at the beginning. 

Let $r_0^i>0$ be the smallest round number in which the node with ID equal to $i$ transmits,
provided $s$ is the only node from the gadget (possibly) sending messages before.
Let $r_0$ be the smallest round number $r>0$ such that $\{i\,|\, r_0^i=r\}\neq\emptyset$.
If $|\{i\in I_0\,|\, r_0^i=r_0\}|=1$, we assign the IDs $j,k$ to $v_0,v_1$, where $j$
is the only element of  $\{i\in I_0\,|\, r_0^i=r_0\}$ and $k$ is an arbitrary
element of $I_0$ different from $j$. 
If $|\{i\in I_0\,|\, r_0^i=r_0\}|>1$, we assign the IDs $j,k$ to $v_0,v_1$, where $j\neq k$
are arbitrary elements of $\{i\in I_0\,|\, r_0^i=r_0\}$. 
Then, we set $I_1=I_0\setminus\{j,k\}$.
Let $u$ be a node with ID in $I_1$ located on a position of one of nodes
from $\{v_2,\ldots,v_{\Delta+1}\}$.
Then: 
(i)~$u$ does not transmit in rounds $r\le r_0$
in which less than two nodes from $\{v_0,v_1\}$ transmit;
(ii)~the feedback which $u$ gets from the 
communication channel until the round $r_0$
does not depend on the actual location of $u$ (by (i) and Fact~\ref{f:core}.1).

Now, inductively, assume that the IDs of $v_0,v_1,\ldots,v_{2a-1}$, the round
number $r_a\ge a$ and $I_a\subset I$ of size $\ge \Delta+4-2a$ are fixed for $a>0$
such that, for each $u$ with ID in $I_a$ located in a position of one of nodes
 $v_{2a},\ldots,v_{\Delta+1}$, the following holds:
(i)~$u$ does not transmit in rounds $r\le r_a$
in which less than two nodes from $\{v_0,\ldots,v_{2a-1}\}$ transmit;
(ii)~the feedback which $u$ gets from the 
communication channel until the round $r_a$
does not depend on the actual location of $u$. 

Now, we assign IDs to $v_{2a}$ and $v_{2a+1}$, set $r_{a+1}>r_a$, and $I_{a+1}$.  
Let $r_{a+1}^i>r_a$ for $i\in I_a$ be the smallest round number larger than $r_a$ in which
the node with ID equal to $i$ transmits, provided the above assumption (i) and (ii) are
satisfied and the only transmitters in rounds $r>r_a$ belong to $\{s,v_0,\ldots,v_{2a-1}\}$.
Let $r_{a+1}$ be the smallest round number $r>r_a$ such that $\{i\,|\, r_{a+1}^i=r\}\neq\emptyset$.
If $|\{i\in I_a\,|\, r_{a+1}^i=r_{a+1}\}|=1$, we assign the IDs $j,k$ to $v_{2a},v_{2a+1}$, where $j$
is the only element of  $\{i\in I_a\,|\, r_{a+1}^i=r_{a+1}\}$ and $k$ is an arbitrary
element of $I_a$ different from $j$. 
If $|\{i\in I_a\,|\, r_{a+1}^i=r_{a+1}\}|>1$, we assign the IDs $j,k$ to $v_{2a},v_{2a+1}$, where $j\neq k$
are arbitrary elements of $\{i\in I_a\,|\, r_{a+1}^i=r_{a+1}\}$. 
Then, we set $I_{a+1}=I_a\setminus\{j,k\}$.

As one can see, the choice of IDs for $v_{2a}$ and $v_{2a+1}$, 
the value of $r_{a+1}>r_a$, and $I_{a+1}$
of size $\ge \Delta+4-2(a+1)$ guarantee that, 
for each $u$ with ID in $I_{a+1}$ located in a position of one of the nodes
 $v_{2a+2},\ldots,v_{\Delta+1}$, the following holds:
(i)~$u$ does not transmit in rounds $r\le r_{a+1}$
in which less than two nodes from $\{v_0,\ldots,v_{2a+1}\}$ transmit;
(ii)~the feedback which $u$ gets from the 
communication channel until the round $r_{a+1}$
does not depend on the actual location of $u$. 

By assigning IDs in the above way we assure that 
$v_{\Delta+1}$ is not the unique transmitter in the gadget core in $\Omega(\Delta)$ rounds after the first message from $s$, thus a message is delivered to $t$ after $\Omega(\Delta)$ rounds.

\comment{
We can assume that all nodes from the set $\{v_0,\ldots,v_{\Delta+1}\}$ receive the first message
transmitted by $s$, thanks to the fact that interference (from outside) inside the gadget's core is bounded by $\nu$ and no node from the core of the gadget is active
before receiving a message from $s$.
Thus, reception of a message from $s$ does not give information which would help the nodes $v_0,\ldots,v_{\Delta+1}$ determine their own positions in the sequence
$v_0,\ldots,v_{\Delta+1}$.

Let $T^{(0)}$ be the schedule describing performance of nodes (IDs) from
the core of the gadget starting in the 
round in which nodes from the core of the gadget receive the first message from $s$,
provided no other message has been received by them.
%
%
According to this definition, a node located in the core with ID equal to $i$ transmits in round $r$ after reception of the first message from $s$ (under the assumptions mentioned before) iff $T^{(0)}_{i,r}=1$. 

Let $r_0$ be the first round in which two nodes (IDs) would transmit (if they were located in the core). That is, $_0$ is the smallest $r$ such that there exist $i \neq j$ with $T^{(0)}_{i,r}=T^{(0)}_{j,r}=1$. We assign such IDs $i,j$ to nodes $v_0,v_1$. 
Observe that, according to the choice of $r_0$ there were at most one prospective transmitter among \emph{all} possible IDs (not assigned yet) for each round $r<r_0$. We remove them from the set $I$ of available IDs. 
During the whole process we will discard at most $\Delta$ IDs for being unique transmitter, so there is no reason to worry about the size of the set of identifiers.

Observe that, in some following rounds the nodes $v_0,v_1$ may be unique transmitters in the core, and this might affect the behavior of nodes $v_2,...,v_{\Delta+1}$. 
Thus, we construct new schedule $T^{(1)}$ which takes behavior of $v_0$ and $v_1$ with assigned IDs $i,j$ into account. However, the distances between nodes and the limit  
$\nu$ on external interference guarantee that, if some node from $\{s,v_0,v_1\}$ transmits 
a message, the transmitted message is received either by \emph{all} or \emph{none} of $v_2,...v_{\Delta+1}$. 
%
\footnote{We take care of this property, since the situation where nodes $v_0,....,v_k$ receive a message from, say $s$, and $v_{k+1},...,v_{\Delta+1}$ do not, might possibly help the algorithm in breaking the symmetry faster.
Similarly, we want to avoid the scenario that a node $v_i$ with determined ID transmits uniquely its message is received by some nodes $v_k,...,v_l$ and not by $v_{l+1},...,v_{\Delta+1}$ 
(this would help in breaking the symmetry among nodes $v_k,\ldots,v_{\Delta+1}$). 
This subtle issue does not occur in models without distant interference, e.g. radio networks, where one can join gadgets one by one, and achieve the lower bound $\Omega(D\Delta)$.}

Now, we look at the schedule $T^{(1)}$.
That is the fact that (all or none from $v_2,...$) nodes can receive messages from nodes $s,v_1,v_2$.
We perform similarly with $T^{(1)}$, and look for the first round $r_1$ when there is a collision, i.e., there are at least two IDs of nodes which transmit if they are located
on positions from the set $\{v_2,\ldots,v_{\Delta+1}\}$. We assign the colliding IDs to nodes $v_2$ and $v_3$. 
We also discard transmitters/IDs that are unique transmitters between round $r_0$ and $r_1$ from the set of available identifiers. Then, we continue the process analogously for $T_{(2)}$ and so on.

If at some point we already know that the algorithm needs $\Omega(\Delta)$ rounds 
(that is, when $r_k = \Omega(\Delta)$) to deliver
a message to $t$, we stop.
Observe that, in such situations, \emph{any} choice of available identifiers from $I$ is good for nodes with unfixed IDs (this is because we discarded IDs that were unique transmitters before $r_k$, so 
$v_{\Delta+1}$ will not be able to transmit to $t$ before $r_k$ anyway).

If we assigned all IDs during the process, we know that $v_{\Delta+1}$ is not a unique transmitter in the gadget core in $\Omega(\Delta)$ rounds after the first message from $s$, thus a message is not delivered to $t$ in $\Omega(\Delta)$ rounds.
}
\end{proof}
Lemma~\ref{l:gadget} gives the lower bound $\Omega(\Delta)$ for the global broadcast.
In order to extend the lower bound to take the network diameter into account, we build a network by joining gadgets sequentially. 
However, the following problem appears in such approach. It is vital for our argument that, in each round, either all nodes from the core of a gadget receive a fixed message or none of them receives any message. If the network consists of more than one gadget, then the interference from other parts of the network might cause that only a subset of the nodes from the core receive a message.
This in turn might help to break symmetry and resolve contention in the core in $o(\Delta)$ rounds. 

To overcome the presented problem, we place ``a buffer zone'' between two consecutive gadgets.
The ``buffer zone'' is a long path of nodes mitigating the interference between two consecutive gadgets (see Fig.~\ref{f:siec}).
In this way we can bound the maximal interference coming from the outside of the gadget to $\nu$
(where $\nu$ is the constant from Lemma~\ref{l:gadget}). 
Given that, we are able to use Lemma~\ref{l:gadget} in networks containing many gadgets.

More precisely, we put a path of $\kappa=\Delta^{1/\alpha}/(1-\eps)$ nodes between two consecutive gadgets to compensate for the interference on a gadget's core caused by potential transmitters from previous gadgets (see Figure~\ref{f:siec} for exact composition of gadgets).

\begin{figure}[H]
  \includegraphics[width=1.0\textwidth]{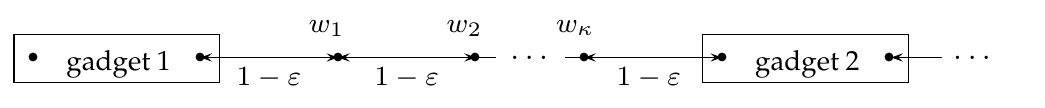}
  \caption{Network of diameter $\Theta(D)$ formed by $D/\kappa$ gadgets, where $\kappa=\frac{\Delta^{1/\alpha}}{1-\varepsilon}$.\label{f:siec}}
\end{figure}

The choice of the value of $\kappa$ follows from the fact that each gadget consists of about $\Delta$ nodes, while we want to limit the interference from neighbouring gadgets to $O(1)$.
The interference from $\Delta$ nodes in distance $d$ from the receiver is equal to $\Delta P / d^\alpha$. This value is $O(1)$ provided $d\ge \Delta^{1/\alpha}$.
Since the path separating two gadgets ``wastes'' about $\Delta^{1/\alpha}$ of the diameter, one can see that we can put as much as $\frac{D}{\Delta^{1/\alpha}}$ of such ``separators'' on a line, 
each followed by a gadget. 
This leads to the lower bound $\Omega(\frac{D}{\Delta^{1/\alpha}}\Delta) = \Omega(D\Delta^{1-1/\alpha})$.
Theorem~\ref{t:lower:bound} follows from
Lemma~\ref{l:gadgets:combined}, which formalizes the above described idea.
\begin{lemma} \labell{l:gadgets:combined}
Consider a network of diameter $D$ and density $\Theta(\Delta)$ formed by $D/\kappa$ gadgets of size $\Delta+4$ located on the line such that consecutive gadgets are ``separated'' by the path of $\kappa$ nodes $w_1,...,w_\kappa$, where $d(w_i,w_{i+1})=1-\eps$ (see Fig.~\ref{f:siec}). 
For any deterministic global broadcast algorithm $\mathcal{A}$, there exists a choice of IDs for the nodes in this network such that 
$\mathcal{A}$ works in time $\Omega(D\Delta^{1-1/\alpha})$.
\end{lemma}
\begin{proof}
First, we prove an auxiliary fact which gives an opportunity to use Lemma~\ref{l:gadget} in multi-gadget networks from Fig.~\ref{f:siec}.
	
	\begin{fact}
		For any gadget $G$ in the network, the interference caused by the nodes outside of the gadget is at most $\nu$ at any point $v_0,v_1,...,v_{\Delta+1}$ from the core of the gadget.
	\end{fact}
	
	\begin{proof}
		All nodes that may interfere are located on the left side of $G$ (i.e., are closer to the source than $G$), we split them in two groups. 
		(Note, that we do not consider the node $s$ from the gadget $G$ as the interferer nor as the core node (c.f.~Fig.~\ref{f:gadget})).
		The first group $T_1$ consists of a path of $\kappa$ nodes that are closest to $G$, the rest of nodes located to the left of the gadget 
		(i.e., the nodes closer to the source than the gadget $G$) belong to the second group $T_2$. 
		The maximal interference caused by nodes from $T_1$ is at most $\sum_{i=1}^{\kappa}\frac{P}{(1-\eps)^\alpha i^\alpha}$.
		
		Now, we bound the interference from $T_2$. 
		First, we estimate the interference from the $k$th gadget to the left 
		of $G$ and from the path separating the $k$th gadget and the $(k-1)$st gadget to the left of $G$. 
		There are $\Delta+\Delta^{1/\alpha}\le 2\Delta$ nodes in this part of the network,
		in distance at least $k\Delta^{1/\alpha}$ from the core of the gadget $G$. This gives 
		the interference 
		smaller than $\frac{2\Delta P}{(k\Delta^{1/\alpha})^\alpha}=2P/k^\alpha$.
		Thus, as there are at most $D/\kappa$ gadgets, the interference from $T_2$ is at most $\sum_{k=1}^{D/\kappa} 2P/k^\alpha$. 
		The total interference $I$ at $G$ is at most 
		$$I\le \sum_{i=1}^{\kappa}\frac{P}{(1-\eps)^\alpha i^\alpha} + \sum_{k=1}^{D/\kappa} 2P/k^\alpha$$ which is $O(1)$ 
		with respect to the parameters $\eps, D$ and $\Delta$. 
		On the other hand, $\nu=\Omega\left(\frac1{\eps^\alpha}\right)$.
		Thus, for sufficiently small values of $\eps$ we have $I\le\nu$. 
	\end{proof}
	
	There are $D/\kappa$ gadgets, each of size $\Delta$ and, by Lemma~\ref{l:gadget}, it takes $\Omega(\Delta)$ to push a message through each gadget. Thus, we have a lower bound of $\Omega(D\Delta/\kappa)=\Omega(D\Delta^{1-1/\alpha})$.
\end{proof}

}

\subsection{Proof of Theorem~\ref{t:lower:bound}}
\lowerboundproof

%% file: appendix1.tex
\subsection{Proofs for Section~\ref{s:preliminaries}}
\lemmaApp{l:density:close}{\ldensityclose}
\begin{proof}
The item 2.\ follows directly from the definitions of a dense cluster, the
function $\chi$ and $d_{\Gamma,r}$.

For item 1., assume that $\cB=B(x,1)$ is dense. Then, there are at least
$\Gamma/2$ nodes in $\cB$ and therefore the smallest distance between nodes located in $\cB$ is at most $d_{\Gamma,1}$. Let $u_0,v_0\in \cB$ be a pair
of nodes in the smallest distance in $\cB$, $d=d(u_0,v_0)= \zeta d_{\Gamma,1}$
for $0\le \zeta\le 1$. If $d(u',v')\ge d/2$ for each 
$u',v'\in B(u_0,\zeta)\cup B(v_0,\zeta)\subset B(u,\zeta+d)$
then $u_0,v_0$ is a close pair.
Otherwise, let $u_1, v_1$ be a pair of nodes in $B(u_0,\zeta+d)$
in distance $\le d/2$. In this way, we can build a sequence of 
different pairs $(u_i,v_i)$ for $i\ge 0$ such that $d(u_i,v_i)\le d/2^i$,
$$u_i,v_i\in B(x, 1+ (\zeta+d)\cdot\sum_{i\ge 0}1/2^i)\subset B(x,1+2(\zeta+d))\subset B(x,5).$$
Eventually, a close pair
$u_j,v_j$ will appear in the sequence, as the number of nodes in $X$ is finite.
And, $u_j,v_j\in B(x,5)$.
\end{proof}

\subsection{Proofs for Section~\ref{s:comb:tools:sinr}}
The proofs of technical lemmas given in this section rely on the polynomial
attenuation of signals with the power $\alpha>2$.

First, we give an estimation of the interference coming from
a set of ``distant nodes'' with limited density.
\begin{proposition}\labell{interEST}
	Let $\calT$ be a set of transmitting nodes, let $v$ be a point on the plane
	and let $x$ be a positive natural number such that
	any ball of radius $r$ contains at most $\delta$ elements of $\calT$
	and $B(v,xr)$ does not contain elements of $\calT$. 
	Then, the overall strength $I(v)=\sum_{u\in\calT}\mathcal{P}\dist(u,v)^{-\alpha}$ 
	of signals from the set $\calT$ received at $v$ is
	$O\left(\frac{\mathcal{P}r^{-\alpha}x^{-\alpha+2}\delta}{\alpha-2}\right)$.
\end{proposition}

\begin{proof} We can estimate $I(v)$ as
	\[I(v)=\sum_{u\in {\cal T}} \mathcal{P} \dist(u,v)^{-\alpha}\leq
	\sum_{i\geq x}|{\cal T}\cap \left(B(v,(i+1)r)\setminus B(v,ir)\right)|\cdot \mathcal{P}\cdot (ir)^{-\alpha}=O\left(\mathcal{P}\delta r^{-\alpha}\sum_{i\geq x} i^{-\alpha+1}\right).
	\]
The last equality above stems from the fact that $(B(v,(i+1)r)\setminus B(v,ir)$
is included in $O(i)$ balls of radius $r$, each such ball contains at most $\delta$ nodes
and all those nodes are in distance $\ge ir$ from $v$.
	By bounding the sum $\sum_{i\geq x} i^{-\alpha+1}$
	by the integral $\int_x^\infty t^{-\alpha+1} dt=\frac{x^{-\alpha+2}}{\alpha-2}$ we get
	\[I(v)=O\left(\frac{ \mathcal{P}r^{-\alpha}x^{-\alpha+2}\delta}{\alpha-2}\right).
	\]
\end{proof}
\comment{First, we establish a bound on the interference in a given point $v$ caused by a set of transmitters $\cal T$ with limited density, provided no element of $\cal T$ is
close to $v$.}

\lemmaApp{l:cons:dens}{\lconsdens}
\begin{proof}
For each point $v$, we  
prove that 
each node in distance at most $1-\eps$ from $v$ transmits
a message received at $v$
if we set $\mathbf{L_\gamma}$ to be a $(N,k)$-ssf for some $k$ which
depends on: $\gamma$, $\eps$ and the SINR parameters $\alpha>2$, $\beta$, $\cal N$.
Let $v_i$ denote the $i$-th closest node to $v$ (where $v_1=v$).
According to our assumptions, only the nodes $v_1,\ldots,v_{\gamma}$ 
can be in distance $\leq 1-\eps$ from $v$. Thus, let $\gamma'\le \gamma$ be
the largest $i$ such that $d(v,v_i)\leq 1-\eps$.
The 
 inequality
\[\beta\leq \frac{\mathcal{P}\dist(v,v_i)^{-\alpha}}{I(v)+{\cal N}},
\]
guarantees that $v$ receives a message transmitted by $v_i$ for $i\leq \gamma'$ in some round,
where $$I(v)=\sum_{u\in {\cal T}\setminus\{v_i\}} \mathcal{P} \dist(u,v)^{-\alpha},$$
$\calT\subset X$ is the set of nodes transmitting in that round, and
$v_i\in \calT$.
Thus, as $d(v,v_i)\le 1-\eps$, $v$ receives the message from $v_i$ if
$$I(v)\leq \mathcal{P}(1-\eps)^{-\alpha}/\beta-{\cal N}={\cal N}\beta(1-\eps)^{-\alpha}/\beta-{\cal N}={\cal N}((1-\eps)^{-\alpha})=O(1).$$ 
%
%
We can bound $I(v)$ as needed using Prop.~\ref{interEST} (with $r=1$),
by choosing constant $x$ large enough (depending on $\eps$
and SINR parameters) and assuring that there is no transmitter (except of $v_i$) in $B(v,x)$. 
Let $k_\gamma$ be the maximal number of nodes in a ball of radius $x$. Observe that $k_\gamma=O(x^2\cdot \gamma)=O(1)$ since $x$ is a constant.
Let $\mathbf{L}_\gamma$ be a $(N,k_\gamma)-$ssf.
Then, for each $i\in[\gamma']$, there is a round in $\mathbf{L}_\gamma$ in which $v_i$ 
is the only transmitter 
in $B(v,x)$. Thus, each neighbor of $v$ transmits in a round where it is the unique transmitter in $B(v,x)$ and $v$ can hear it. 

Thus, the length of the schedule $\mathbf{L}_\gamma$ is $O(\gamma^2\log N)=O(\log N)$, \cite{PoratR11}.

\end{proof}



\lemmaApp{l:close}{
\lclose
}

\begin{proof}

Let $v,u$ be a close pair and let $d(u,v)=d=\zeta d_{\Gamma,1}$ for
$\zeta\in(0,1]$. Moreover, let $\calT$ be the set of transmitters in a round and let $I(u)=\sum_{w\in {\cal T}\setminus \{v\}} \mathcal{P}/d(w,u)^{-\alpha}$ be the interference at $u$ caused by other nodes. 
The following condition 
\[\beta\leq \frac{\mathcal{P}/d^{-\alpha}}{I(u)+{\cal N}},
\]
guarantees that $u$ receives the message transmitted by $v$.
In order to satisfy the above inequality, it is
sufficient that $I(u)\le \mathcal{P}d^{-\alpha}/(2\beta)$ and ${\cal N}\le \mathcal{P}d^{-\alpha}/(2\beta)$.

The latter inequality is equivalent to
\begin{equation}\label{e:close:interference0}
d\le \left(\frac{2{\cal N}\beta}{\calP}\right)^{-1/\alpha}=2^{-1/\alpha}\le {\Gamma}^{-1/\alpha},
\end{equation}
where the equality follows from the assumption
$\calP={\cal N}\beta$. Thus, the latter condition can be guaranteed 
if $d_{\Gamma,r}\le 2^{-1/\alpha}$, which holds for each sufficiently large $\Gamma$ (recall that
$d\le d_{\Gamma,1}$.)

In order to estimate $I(u)$, we choose $x\in\NAT$ such that $xd\le \zeta$ and split $\calT$ into
$\calT_0=\calT\cap B(u,xd)$, $\calT_1=\calT\cap(B(u,\zeta)-B(u,xd))$ and
$\calT_2=\calT-(\calT_0\cup\calT_1)$. Analogously, $I(u)$ is split into
$I_0, I_1, I_2$, where $I_j=\sum_{w\in \calT_j}\calP\dist(w,u)^{-\alpha}$
for $j\in[0,2]$. The exact value of $x$ will be determined later.

Let $A'$ be the set of all nodes from $A$ located inside $B(u,xd)$.
%
The assumption that $u,v$ is a close pair
guarantees that each pair of nodes
inside $B(u,\zeta)$ is in distance $\ge d/2$. Thus the number of nodes in 
$A'$ is at most $\chi(dx, d/2)=\Theta(x^2)$, since nodes are located on
the plane. 

In the following, we will analyse the scenario that $v$ is the only transmitter
from $A'$. Then $I_0=0$, since $\calT_0\subseteq A'$. In order to prove the lemma,
it is sufficient to prove that 
$I_1\le \mathcal{P}d^{-\alpha}/(4\beta)$ and $I_2\le \mathcal{P}d^{-\alpha}/(4\beta)$.


In order to estimate $I_1$, we apply Prop.~\ref{interEST} with $r=d$,
$\delta=\chi(d,d/2)=O(1)$ and $\calT=\calT_1$, obtaining
%
\begin{equation}\label{e:close:interference1}
I_1=O\left(\frac{\mathcal{P}d^{-\alpha}x^{-\alpha+2}}{\alpha-2}\right).
\end{equation}
In order to estimate $I_2$, we make use of the assumption that each ball of radius $1$ contains at most $\Gamma$ nodes and apply Prop.~\ref{interEST} with $r=\zeta\leq 1$,
$x=1$, $\delta=\Gamma$ and $\calT=\calT_2$, which gives
\begin{equation}\label{e:close:interference2}
I_2=O\left(\frac{\zeta^{-\alpha}\mathcal{P}\Gamma}{\alpha-2}\right).
\end{equation}
The value of $I_1$ can be made smaller than $\mathcal{P}(2d)^{-\alpha}/(4\beta)$ by
taking big enough but constant (i.e., depending only on the model parameters) $x$.
(Recall that constant $x$ would give $|A'|=\kappa=O(x^2)=O(1)$.)

The value of $I_2$ is
$$I_2=O\left(\frac{\zeta^{-\alpha} \mathcal{P}\Gamma}{\alpha-2}\right)$$
while the bound $\mathcal{P}\frac{(2d)^{-\alpha}}{4\beta}$ satisfies
$$\mathcal{P}\frac{(2d)^{-\alpha}}{4\beta}=\Omega\left(\frac{\mathcal{P}}{4\beta}\cdot\left(2\cdot\zeta(1/\Gamma)^{1/2}\right)^{-\alpha}\right)=\Omega(\zeta^{-\alpha}\cdot \Gamma^{\alpha/2}),$$
since $d\le d_{\Gamma,1}=\Theta(\Gamma^{1/2})$.
Thus, $I_2\le \mathcal{P}\frac{(2d)^{-\alpha}}{4\beta}$ if $\Gamma$ is big enough.

Finally, if $\Gamma$ is smaller than the values needed to satisfy (\ref{e:close:interference1}),
or (\ref{e:close:interference2}), or (\ref{e:close:interference0}), we can use
the construction from Lemma~\ref{l:cons:dens}.

\end{proof}

\lemmaApp{l:close:clustered}{
\lcloseclustered
}
\begin{proof}
If $\clid$
is the only nonempty cluster, the result can be obtained by the
same reasoning as in the proof of Lemma~\ref{l:close}. 

In order to take other clusters into account recall that nodes
from $\gamma=O(1)$ clusters might appear in each unit ball,
by 
the definition of $r$-clustering (since the centers of clusters
are in distance $\ge 1-\eps$).
Thus, the number of nodes from all clusters in a unit 
ball is limited by $\gamma\Gamma=O(\Gamma)$.
For a cluster $\clid$ included in $B(x,r)$, let the set $C$ of 
clusters ``in conflict'' with $\clid$ consists of all clusters
(excluding $\clid$) whose nodes (at least one of them) are located in 
$B(x,2r)$. 
There are $\rho=O(1)$ such
clusters.
Then, the interference  
from other clusters (excluding those from $C$) on the nodes of a close pair in $\clid$ can be limited
similarly as in Lemma~\ref{l:close}, thanks to the fact
that the number of nodes in a unit ball is $O(\Gamma)$,
the distance between nodes of a close pair is $O(1/\Gamma^{1/2})$
and $\alpha>2$. 
\end{proof}

\comment{
\section{Proofs for Section~\ref{s:sparsification}}

\lemmaApp{l:sparsification:unclustered}{\lsparsuncluster}
\plsparsuncluster

\section{The figure for Section~\ref{s:basic:SINR}}
\figelim
}

%% file: main.bbl
\begin{thebibliography}{10}

\bibitem{AlonSpencer-book}
N.~Alon and J.~H. Spencer.
\newblock {\em The Probabilistic Method}.
\newblock Wiley Publishing, 4th edition, 2016.

\bibitem{AronovK15}
B.~Aronov and M.~J. Katz.
\newblock Batched point location in {SINR} diagrams via algebraic tools.
\newblock In {\em Automata, Languages, and Programming - 42nd International
  Colloquium, {ICALP} 2015, Kyoto, Japan, July 6-10, 2015, Proceedings, Part
  {I}}, pages 65--77, 2015.

\bibitem{Bar-YehudaGI92}
R.~Bar{-}Yehuda, O.~Goldreich, and A.~Itai.
\newblock On the time-complexity of broadcast in multi-hop radio networks: An
  exponential gap between determinism and randomization.
\newblock {\em J. Comput. Syst. Sci.}, 45(1):104--126, 1992.

\bibitem{BarenboimSIROOCCO15}
L.~Barenboim and D.~Peleg.
\newblock Nearly optimal local broadcasting in the {SINR} model with feedback.
\newblock In {\em Structural Information and Communication Complexity - 22nd
  International Colloquium, {SIROCCO} 2015, Montserrat, Spain, July 14-16,
  2015, Post-Proceedings}, pages 164--178, 2015.

\bibitem{ChlebusGGPR00}
B.~S. Chlebus, L.~Gasieniec, A.~Gibbons, A.~Pelc, and W.~Rytter.
\newblock Deterministic broadcasting in unknown radio networks.
\newblock In {\em Proceedings of the Eleventh Annual {ACM-SIAM} Symposium on
  Discrete Algorithms, January 9-11, 2000, San Francisco, CA, {USA.}}, pages
  861--870, 2000.

\bibitem{ClementiMS01}
A.~E.~F. Clementi, A.~Monti, and R.~Silvestri.
\newblock Selective families, superimposed codes, and broadcasting on unknown
  radio networks.
\newblock In S.~R. Kosaraju, editor, {\em SODA}, pages 709--718. ACM/SIAM,
  2001.

\bibitem{ClementiMS03}
A.~E.~F. Clementi, A.~Monti, and R.~Silvestri.
\newblock Distributed broadcast in radio networks of unknown topology.
\newblock {\em Theor. Comput. Sci.}, 302(1-3):337--364, 2003.

\bibitem{CornejoK10}
A.~Cornejo and F.~Kuhn.
\newblock Deploying wireless networks with beeps.
\newblock In {\em Distributed Computing, 24th International Symposium, {DISC}
  2010, Cambridge, MA, USA, September 13-15, 2010. Proceedings}, pages
  148--162, 2010.

\bibitem{CzumajRytter-FOCS-03}
A.~Czumaj and W.~Rytter.
\newblock Broadcasting algorithms in radio networks with unknown topology.
\newblock In {\em FOCS}, pages 492--501. IEEE Computer Society, 2003.

\bibitem{DaumGKN13}
S.~Daum, S.~Gilbert, F.~Kuhn, and C.~C. Newport.
\newblock Broadcast in the ad hoc {SINR} model.
\newblock In Y.~Afek, editor, {\em Distributed Computing - 27th International
  Symposium, {DISC} 2013, Jerusalem, Israel, October 14-18, 2013. Proceedings},
  volume 8205 of {\em Lecture Notes in Computer Science}, pages 358--372.
  Springer, 2013.

\bibitem{DeMarco-SICOMP-10}
G.~DeMarco.
\newblock Distributed broadcast in unknown radio networks.
\newblock {\em SIAM J. Comput.}, 39(6):2162--2175, 2010.

\bibitem{EmekGKPPS09}
Y.~Emek, L.~Gasieniec, E.~Kantor, A.~Pelc, D.~Peleg, and C.~Su.
\newblock Broadcasting in udg radio networks with unknown topology.
\newblock {\em Distributed Computing}, 21(5):331--351, 2009.

\bibitem{EmekKP16}
Y.~Emek, E.~Kantor, and D.~Peleg.
\newblock On the effect of the deployment setting on broadcasting in euclidean
  radio networks.
\newblock {\em Distributed Computing}, 29(6):409--434, 2016.

\bibitem{ForsterSW14}
K.~F{\"{o}}rster, J.~Seidel, and R.~Wattenhofer.
\newblock Deterministic leader election in multi-hop beeping networks -
  (extended abstract).
\newblock In {\em Distributed Computing - 28th International Symposium, {DISC}
  2014, Austin, TX, USA, October 12-15, 2014. Proceedings}, pages 212--226,
  2014.

\bibitem{FuchsW13}
F.~Fuchs and D.~Wagner.
\newblock On local broadcasting schedules and {CONGEST} algorithms in the
  {SINR} model.
\newblock In {\em Algorithms for Sensor Systems - 9th International Symposium
  on Algorithms and Experiments for Sensor Systems, Wireless Networks and
  Distributed Robotics, {ALGOSENSORS} 2013, Sophia Antipolis, France, September
  5-6, 2013, Revised Selected Papers}, pages 170--184, 2013.

\bibitem{GoussevskaiaMW08}
O.~Goussevskaia, T.~Moscibroda, and R.~Wattenhofer.
\newblock Local broadcasting in the physical interference model.
\newblock In M.~Segal and A.~Kesselman, editors, {\em DIALM-POMC}, pages
  35--44. ACM, 2008.

\bibitem{Gudm}
H.~Gudmundsdottir, E.~I. {\'{A}}sgeirsson, M.~H.~L. Bodlaender, J.~T. Foley,
  M.~M. Halld{\'{o}}rsson, and Y.~Vigfusson.
\newblock Extending wireless algorithm design to arbitrary environments via
  metricity.
\newblock In {\em 17th {ACM} International Conference on Modeling, Analysis and
  Simulation of Wireless and Mobile Systems, MSWiM'14, Montreal, QC, Canada,
  September 21-26, 2014}, pages 275--284, 2014.

\bibitem{HalldorssonHL15}
M.~M. Halld{\'{o}}rsson, S.~Holzer, and N.~A. Lynch.
\newblock A local broadcast layer for the {SINR} network model.
\newblock In {\em Proceedings of the 2015 {ACM} Symposium on Principles of
  Distributed Computing, {PODC} 2015, Donostia-San Sebasti{\'{a}}n, Spain, July
  21 - 23, 2015}, pages 129--138, 2015.

\bibitem{HM12}
M.~M. Halld{\'o}rsson and P.~Mitra.
\newblock Towards tight bounds for local broadcasting.
\newblock In F.~Kuhn and C.~C. Newport, editors, {\em FOMC}, page~2. ACM, 2012.

\bibitem{HalldorssonT15}
M.~M. Halld{\'{o}}rsson and T.~Tonoyan.
\newblock How well can graphs represent wireless interference?
\newblock In {\em Proceedings of the Forty-Seventh Annual {ACM} on Symposium on
  Theory of Computing, {STOC} 2015, Portland, OR, USA, June 14-17, 2015}, pages
  635--644, 2015.

\bibitem{hobbs2012deterministic}
N.~Hobbs, Y.~Wang, Q.-S. Hua, D.~Yu, and F.~C. Lau.
\newblock Deterministic distributed data aggregation under the sinr model.
\newblock In {\em Theory and Applications of Models of Computation}, pages
  385--399. Springer Berlin Heidelberg, 2012.

\bibitem{JK-DISC-12}
T.~Jurdzinski and D.~R. Kowalski.
\newblock Distributed backbone structure for algorithms in the sinr model of
  wireless networks.
\newblock In M.~K. Aguilera, editor, {\em DISC}, volume 7611 of {\em Lecture
  Notes in Computer Science}, pages 106--120. Springer, 2012.

\bibitem{JurdzinskiK16a}
T.~Jurdzinski and D.~R. Kowalski.
\newblock Wake-up problem in multi-hop radio networks.
\newblock In {\em Encyclopedia of Algorithms}, pages 2352--2354. Springer,
  2016.

\bibitem{JKRS13}
T.~Jurdzinski, D.~R. Kowalski, M.~Rozanski, and G.~Stachowiak.
\newblock Distributed randomized broadcasting in wireless networks under the
  sinr model.
\newblock In {\em DISC}, pages 373--387, 2013.

\bibitem{JKRS14}
T.~Jurdzinski, D.~R. Kowalski, M.~Rozanski, and G.~Stachowiak.
\newblock On the impact of geometry on ad hoc communication in wireless
  networks.
\newblock In M.~M. Halld{\'{o}}rsson and S.~Dolev, editors, {\em {ACM}
  Symposium on Principles of Distributed Computing, {PODC} '14, Paris, France,
  July 15-18, 2014}, pages 357--366. {ACM}, 2014.

\bibitem{JKS13}
T.~Jurdzinski, D.~R. Kowalski, and G.~Stachowiak.
\newblock Distributed deterministic broadcasting in uniform-power ad hoc
  wireless networks.
\newblock In L.~Gasieniec and F.~Wolter, editors, {\em FCT}, volume 8070 of
  {\em Lecture Notes in Computer Science}, pages 195--209. Springer, 2013.

\bibitem{JRS17}
T.~Jurdzinski, M.~Rozanski, and G.~Stachowiak.
\newblock Token traversal in ad hoc wireless networks via implicit carrier
  sensing.
\newblock SIROCCO 2017, 2017.

\bibitem{KantorFOCS15}
E.~Kantor, Z.~Lotker, M.~Parter, and D.~Peleg.
\newblock The minimum principle of {SINR:} {A} useful discretization tool for
  wireless communication.
\newblock In {\em {IEEE} 56th Annual Symposium on Foundations of Computer
  Science, {FOCS} 2015, Berkeley, CA, USA, 17-20 October, 2015}, pages
  330--349, 2015.

\bibitem{Kesselheim11}
T.~Kesselheim.
\newblock A constant-factor approximation for wireless capacity maximization
  with power control in the sinr model.
\newblock In D.~Randall, editor, {\em SODA}, pages 1549--1559. SIAM, 2011.

\bibitem{KowalskiPelc2003}
D.~R. Kowalski and A.~Pelc.
\newblock Broadcasting in undirected ad hoc radio networks.
\newblock In {\em Proceedings of the Twenty-second Annual Symposium on
  Principles of Distributed Computing}, PODC '03, pages 73--82, New York, NY,
  USA, 2003. ACM.

\bibitem{OgiermanRSSZ14}
A.~Ogierman, A.~W. Richa, C.~Scheideler, S.~Schmid, and J.~Zhang.
\newblock Competitive {MAC} under adversarial {SINR}.
\newblock In {\em 2014 {IEEE} Conference on Computer Communications, {INFOCOM}
  2014, Toronto, Canada, April 27 - May 2, 2014}, pages 2751--2759, 2014.

\bibitem{PoratR11}
E.~Porat and A.~Rothschild.
\newblock Explicit nonadaptive combinatorial group testing schemes.
\newblock {\em {IEEE} Transactions on Information Theory}, 57(12):7982--7989,
  2011.

\bibitem{RichaS16}
A.~W. Richa and C.~Scheideler.
\newblock Jamming-resistant {MAC} protocols for wireless networks.
\newblock In {\em Encyclopedia of Algorithms}, pages 999--1002. Springer, 2016.

\bibitem{SchneiderW08}
J.~Schneider and R.~Wattenhofer.
\newblock A log-star distributed maximal independent set algorithm for
  growth-bounded graphs.
\newblock In {\em Proceedings of the Twenty-Seventh Annual {ACM} Symposium on
  Principles of Distributed Computing, {PODC} 2008, Toronto, Canada, August
  18-21, 2008}, pages 35--44, 2008.

\bibitem{YuHWL12}
D.~Yu, Q.~Hua, Y.~Wang, and F.~C.~M. Lau.
\newblock An o(log n) distributed approximation algorithm for local
  broadcasting in unstructured wireless networks.
\newblock In {\em {IEEE} 8th International Conference on Distributed Computing
  in Sensor Systems, {DCOSS} 2012, Hangzhou, China, 16-18 May, 2012}, pages
  132--139, 2012.

\end{thebibliography}
